\documentclass[final,3p,12pt,times,authoryear]{elsarticle}

\usepackage{amsthm,amsmath,natbib,amsfonts,graphicx,float,dsfont,amssymb}
\RequirePackage[colorlinks,citecolor=blue,urlcolor=blue]{hyperref}

\newtheorem{thm}{Theorem}
\newtheorem{prop}{Proposition}
\newtheorem{Corollary}{Corollary}
\newtheorem{lemma}{Lemma}
\newtheorem{remark}{Remark}
\DeclareMathOperator{\re}{\mathbb{R}}
\DeclareMathOperator{\na}{\mathbb{N}}
\newcommand{\abs}[1]{\left|#1\right|}
\newcommand{\E}{\mathbb{E}}
\newcommand{\ind}{\mathds{1}}

\newcommand{\Prob}{\mathbb{P}}
\newcommand{\F}{\mathcal{F}}

\allowdisplaybreaks

\journal{Mathematics and Computers in Simulation}

\begin{document}

\begin{frontmatter}

\title{Weak Convergence and Optimal Tuning of the Reversible Jump Algorithm}

\author[label1]{Philippe Gagnon\corref{cor1}}
\address[label1]{Department of Statistics, University of Oxford, 24-29 St Giles', Oxford, OX1 3LB, United Kingdom}
\ead{philippe.gagnon@stats.ox.ac.uk}
\cortext[cor1]{Corresponding author}
\author[label2]{Myl\`{e}ne B\'{e}dard}
\address[label2]{D\'{e}partement de math\'{e}matiques et de statistique, Universit\'{e} de Montr\'{e}al, C.P. 6128, Succursale Centre-ville, Montr\'{e}al, QC, H3C 3J7, Canada}
\ead{mylene.bedard@umontreal.ca}
\author[label3]{Alain Desgagn\'{e}}
\address[label3]{D\'{e}partement de math\'{e}matiques, Universit\'{e} du Qu\'{e}bec \`{a} Montr\'{e}al, C.P. 8888, Succursale Centre-ville, Montr\'{e}al, QC, H3C 3P8, Canada}
\ead{desgagne.alain@uqam.ca}

\begin{abstract}
  The reversible jump algorithm is a useful Markov chain Monte Carlo method introduced by \cite{green1995reversible} that allows switches between subspaces of differing dimensionality, and therefore, model selection. Although this method is now increasingly used in key areas (e.g.\ biology and finance), it remains a challenge to implement it. In this paper, we focus on a simple sampling context in order to obtain theoretical results that lead to an optimal tuning procedure for the considered reversible jump algorithm, and consequently, to easy implementation. The key result is the weak convergence of the sequence of stochastic processes engendered by the algorithm. It represents the main contribution of this paper as it is, to our knowledge, the first weak convergence result for the reversible jump algorithm. The sampler updating the parameters according to a random walk, this result allows to retrieve the well-known $0.234$ rule for finding the optimal scaling. It also leads to an answer to the question: ``with what probability should a parameter update be proposed comparatively to a model switch at each iteration?''
\end{abstract}


\end{frontmatter}

\section{Introduction}\label{sec_intro_weak}
Markov chain Monte Carlo (MCMC) methods provide algorithms to generate Markov chains having an invariant measure that corresponds to the distribution with respect to which we are interested in computing integrals.
The implementation of such samplers usually requires the specification of proposal kernels. In the Metropolis-Hastings (MH) algorithm (\cite{metropolis1953equation} and \cite{hastings1970monte}), the most commonly used method, the proposal kernel corresponds to a proposal distribution that is used at each step to generate a candidate for the next state of the Markov chain. The candidate is accepted according to a probability function that is provided by the method.
There are no guidelines that are valid in all situations for effectively achieving the specification step, at least no practical ones. The problem is that poor designs of the proposal kernels may lead to inefficient algorithms, resulting in Markov chains that slowly explore their state space, and therefore, inadequate samples (see \cite{peskun1973optimum} and \cite{tierney1998note}). Researchers have therefore focused on finding optimal tuning procedures in special cases.

\cite{roberts1997weak} studied the MH algorithm in the situation where the proposal distribution is a normal centred around the current state of the chain; algorithms of this type are called random walk Metropolis (RWM) algorithms. In their case, the specification step consists in selecting the variance of the normal distribution. This task is however not trivial as small variances lead to tiny movements of the Markov chain, while large variances induce high rejection rates of candidates. Assuming that the algorithm is used to sample from a distribution of $n$ independent and identically distributed (i.i.d.) random variables, the authors proved the existence of an asymptotically optimal variance for the random walk as $n\rightarrow\infty$ through a weak convergence argument. They also provided the following simple strategy to identify the (asymptotically) optimal scaling: tune the variance so that the acceptance rate of candidates is $0.234$. This leads to a straightforward implementation of the algorithm. A lot of research has been carried out to generalise this result to more elaborate target distributions (e.g.\ \cite{roberts2001optimal}, \cite{neal2006optimal}, \cite{bedard2007weak}, \cite{bedard2008optimal}, \cite{beskos2009optimal}, \cite{bedard2012scaling}, \cite{mattingly2012diffusion} and \cite{beskos2013optimal}). What has been observed is that the $0.234$ rule is robust, in the sense that it often leads to scalings that are (almost) optimal even when the i.i.d. assumption is violated.

A flaw of MH algorithms is that they do not allow subspace switches, i.e.\ switching from the parameter space of a model to that of another model. This gap was corrected by \cite{green1995reversible} with the introduction of the reversible jump (RJ) algorithm. This method has a tremendous potential because of its capability to deliver information on both the ``good'' models and their parameters from a single output. For instance, \cite{richardson1997bayesian} used it to estimate the number of components of mixtures, and their parameters, simultaneously. The price to pay is that the proposal kernel that has to be specified in order to implement this algorithm is much more complicated. In this paper, we study the RJ algorithm in the same mindset as \cite{roberts1997weak}; we aim at providing guidelines to users and open new research directions towards an automatic RJ algorithm.

Existing research on the RJ algorithm has mainly focused on ways to facilitate subspace switches (e.g.\ \cite{brooks2003efficient}, \cite{hastie2005towards}, \cite{al2004improving} and \cite{karagiannis2013annealed}), and therefore, the exploration of the entire state space. The main drawback of the proposed approaches is the difficulty to implement them. This justifies the need for practical guidelines that will promote accessibility of the RJ algorithm. As a first step towards automated implementation of this method, we focus on a simple sampling context in order to obtain weak convergence results. The context is defined in Section~\ref{sec_context_weak}. We consider the natural generalisation of the target distribution in \cite{roberts1997weak}, i.e.\ the case where the parameters are conditionally i.i.d. given any Model $K$. We assume additional structure on the target to decrease the complexity of the proposal kernel required for the implementation. In particular, it makes the main strategies proposed by the authors mentioned above unnecessary (the functions that are used to map the parameters of a model to those of another are not difficult to identify in our context). We assume that Model $K+1$ is comprised of one more parameter than Model $K$, and that when switching from Model $K$ to Model $K+1$, the parameters that are common in both models retain their distribution. Random walks are used to explore the model and parameter spaces. This context can correspond to variable selection in linear regression in the situation where the scale parameters of the error terms are known and the variables are orthogonal. Although the assumptions are strong, we believe that the tuning rules found under them are robust. We also believe that this is a good starting point for proving weak convergence results for more general target distributions and RJ algorithms.

The weak convergence of the sequence of stochastic processes engendered by the algorithm towards diffusion processes is established in Section~\ref{sec_problem_weak}. We explain in Section~\ref{sec_optimisation} that this result allows to establish the validity of the $0.234$ rule of \cite{roberts1997weak}, when the acceptance rate is computed considering only the iterations in which it is proposed to update the parameters. This result also allows finding the probability according to which a model switch (and therefore a parameter update) should be proposed at each iteration. Essentially, the poorer is the design of the proposal distribution generating  candidates for the additional parameter when switching from Model $K$ to Model $K+1$, the higher should be the probability to propose model switches. This is intuitive given that a poor design leads to poor candidates, and therefore, high rejection rates. The result indicates that the optimal probability for proposing parameter updates decreases as $1/\sqrt{A}$, where $1/A$ is equal (up to a known constant) to the probability to propose and accept model switches. Therefore, given that (in our case) the optimal probability for proposing parameter updates when $A=2$ is $0.415$, this provides a rule for constructing the proposal distribution for the different movement types. The proposed rules for tuning the RJ algorithm are described in detail in Section~\ref{sec_optimal_impl}. The advantage of considering a simple setting is that we obtain explicit solutions. In Section~\ref{sec_optimal_impl}, we also present situations in which we conjecture that our results hold. They include, for instance, posterior distributions arising from robust principal component regressions. We show in Section~\ref{sec_simulation_weak} how the design of the RJ algorithm has an impact on the produced samples. We observe that, for moderate dimensions, selecting any probability between $0.2$ and $0.6$ for proposing parameter updates is almost optimal, which seems to be a valid general guideline.


\section{Sampling Context}\label{sec_context_weak}
Let
$$
 \pi_n(k,\mathbf{x}^k)=p_n(k)\prod_{i=1}^{n+k} f(x_i^k)
$$
be the joint posterior distribution of $(K,\mathbf{X}^{K})$, where $K\in\{1,\ldots,\lfloor \sqrt{n}\log n\rfloor\}$ ($\lfloor \cdot \rfloor$ is the floor function), $\mathbf{X}^{K}:=(X_1^{K},\ldots, X_{n+K}^{K})\in\re^{n+K}$, $n\in\{7,8,\ldots\}$, $f$ is a strictly positive one-dimensional probability density function (PDF), and $p_n$ is a probability mass function (PMF) such that $p_n(k)>0$ for all $k\in \{1,\ldots,\lfloor \sqrt{n}\log n\rfloor\}$. We consider $n$ to be an integer greater than or equal to 7 just to avoid technical complications in the proofs. The random variable $K$ represents the model indicator ($K=1$ implies that Model 1 is considered, for instance), and $\mathbf{X}^{K}$ is the parameter vector of Model $K$. Therefore,  $X_1^1,\ldots,X_{n+1}^1$ are the $n+1$ parameters of Model 1, $X_1^2,\ldots,X_{n+2}^2$ are the $n+2$ parameters of Model 2, etc.
As mentioned in the introduction, $\pi_n$ can correspond to the joint posterior of the models and their parameters in variable selection in the situation where the scale parameters of the error terms are known and the variables are orthogonal. In fact, it represents the situation where, in addition, the models comprised of $n$ variables or less have negligible probabilities (which is modelled by $p_n(k)=0$ for $k\leq n$).

The objective is to obtain adequate samples from the joint posterior distribution of $(K,\mathbf{X}^{K})$ through MCMC methods in order to approximate probabilities, expectations, or any other quantity we might be interested in. The following RJ algorithm is applied to approximate a sample from $\pi_n$:
\begin{enumerate}
 \item Initialisation: set $(K(0),\mathbf{X}^{K(0)}(0))$.
 \item[Iteration $m+1$.\hspace{-23.5mm}] \hspace{23.5mm} (Given the current state $(K(m),\mathbf{X}^{K(m)}(m))$.)
 \item Independently generate $U$ and  $U_a$ from a uniform distribution on $(0,1)$.
 \item[3.(a)\hspace{-4.5mm}] \hspace{4.5mm} When $U\leq g(1)$ ($g$ is a PDF such that $g(j)>0$ for all $j\in\{1,2,3\}$), attempt an update of the parameters. Generate $\mathbf{Y}^{K(m)}(m+1)\sim \mathcal{N}(\mathbf{X}^{K(m)}(m),(\ell^2/(n+{K(m)}))\mathcal{I}_{n+K(m)})$, where $\mathbf{Y}^{K(m)}(m+1):=(Y_1^{K(m)}(m+1), \ldots,Y_{n+K(m)}^{K(m)}$ $(m+1))$, $\mathcal{I}_{n+K(m)}$ is the identity matrix of size $n+K(m)$ and $\ell$ is a positive constant. 
 When
  \begin{equation}\label{eqn_prob_update}
  U_a\leq 1 \wedge \frac{\prod_{i=1}^{n+K(m)} f(Y_i^{K(m)}(m+1))}{\prod_{i=1}^{n+K(m)} f(X_i^{K(m)}(m))},
 \end{equation}
 set $(K(m+1),\mathbf{X}^{K(m+1)}$ $(m+1))=(K(m),\mathbf{Y}^{K(m)}(m+1))$.
 \item[3.(b)\hspace{-4.5mm}] \hspace{4.5mm} When $g(1)<U\leq g(1)+g(2)$, attempt adding a parameter to switch from Model $K(m)$ to Model $K(m)+1$. Generate $U(m+1)\sim q$ ($q$ is a strictly positive PDF).
 When
 \begin{equation}\label{eqn_prob_add}
   U_a \leq 1 \wedge \frac{f(U(m+1))p_n(K(m)+1)g(3)}{q(U(m+1))p_n(K(m))g(2)},
  \end{equation}
  set $(K(m+1),\mathbf{X}^{K(m+1)}(m+1))=(K(m)+1,(\mathbf{X}(m)^{K(m)},U(m+1)))$.
 \item[3.(c)\hspace{-4.5mm}] \hspace{4.5mm} When $U>g(1)+g(2)$, attempt withdrawing the last parameter to switch from Model $K(m)$ to Model $K(m)-1$.  
 When  \begin{equation}\label{eqn_prob_remove}
    U_{\text{$a$}} \leq 1 \wedge \frac{q(X_{n+K(m)}^{K(m)}(m))p_n(K(m)-1)g(2)}{f(X_{n+K(m)}^{K(m)}(m))p_n(K(m))g(3)},
 \end{equation}
 set $(K(m+1),\mathbf{X}^{K(m+1)}(m+1))=(K(m)-1,\mathbf{X}^{K(m)-}(m))$, where $\mathbf{X}^{K(m)-}(m)$ is the vector $\mathbf{X}^{K(m)}(m)$ without the last component (more precisely $\mathbf{X}^{K(m)-}(m):=(X_1^{K(m)},\ldots, X_{n+K(m)-1}^{K(m)})$).
 \item[4.] In case of rejection, set $(K(m+1),\mathbf{X}^{K(m+1)}(m+1))=(K(m),\mathbf{X}^{K(m)}(m))$.
 \item[5.] Go to Step 2.
\end{enumerate}

The resulting stochastic process $\{(K(m),\mathbf{X}^K(m)): m\in\na\}$ is a $\pi_n$-irreducible and aperiodic Markov chain. In addition, it is easily shown that this Markov chain satisfies the reversibility condition with respect to $\pi_n$ (see \cite{gagnon2016thesis}), and therefore, that it is ergodic, which guarantees that the Law of Large Numbers holds.

Regularity conditions are now imposed on the different functions. They allow obtaining the weak convergence result stated in Section~\ref{sec_result_weak}.

First, we assume that the usual smoothness conditions on the function $f$ are satisfied: $f\in\mathcal{C}^2(\re)$ (the space of real-valued functions on $\re$ with continuous second derivative), $(\log f(x))'$ is Lipschitz continuous and $\E[((\log f(X))')^4]<\infty$, where the expectation is computed with respect to $f$. This last condition can be replaced by $\E[(f''(X)/f(X))^2]<\infty$, which is slightly stronger. We additionally assume that there exists a constant $A^*\geq 1$ such that
$$
 0<\frac{f}{q}\leq A^* \text{\quad and therefore \quad} \frac{1}{A^*}\leq \frac{q}{f}<\infty.
$$
This condition corresponds to that required for the rejection sampling me\-thod. It ensures that the tails of $q$ are at least as heavy as those of $f$.  A small value for the constant $A^*$ means that $q$ is similar to $f$, and therefore, that it is a good choice of proposal distribution. Note that, when we can directly sample from $f$, we can set $q=f$.

The distribution of $K$, which is $p_n$, also fulfils some conditions. The problem here is that $K$ is a discrete random variable, but the goal is to establish the weak convergence of its associated stochastic process towards a diffusion. A natural way of taking a step in this direction is to assume that the distribution of a suitable transformation of $K$ converges towards one having a PDF. In other words, it is to assume that, in the limit, this transformation of $K$ is a continuous random variable.

Although we could study various distributions $p_n$, we focus on the case where the transformation of $K$ is asymptotically distributed as a standard normal. This represents a natural model to consider that allows obtaining explicit solutions to the optimal tuning problem. The theoretical and experimental results of Sections~\ref{sec_result_weak} and \ref{sec_simulation_weak} shall then be expounded under that framework; it is however underlined in Section~\ref{sec_problem_weak} that the results hold in greater generality, for a larger class of probability mass functions $p_n$.

We assume that $p_n$ has its mode in the middle of the set $\{1,\ldots,\lfloor \sqrt{n}\log n\rfloor\}$ and  is symmetric with respect to this mode. Two distinct cases thus have to be considered: when $\lfloor \sqrt{n}\log n\rfloor$ is even or odd. When $\lfloor \sqrt{n}\log n\rfloor$ is odd, the mode is $(\lfloor \sqrt{n}\log n\rfloor+1)/2$ and we assume that
\begin{align*}
  &p_n(k+1)=a_{k,n}p_n(k), \quad  k\in\{(\lfloor \sqrt{n}\log n\rfloor+1)/2,\ldots,\lfloor \sqrt{n}\log n\rfloor-1\}, \cr
  &p_n(k-1)=a_{k-1,n}p_n(k), \quad  k\in\{2,\ldots,(\lfloor \sqrt{n}\log n\rfloor+1)/2\},
\end{align*}
where $a_{k,n}:=(1-b_{k,n}/\sqrt{n})$ with
$$
 b_{k,n}:=\abs{\frac{k-\lfloor \sqrt{n}\log n\rfloor/2}{\sqrt{n}}}.
$$
Note that $a_{k,n}$ decreases with the distance between $k$ and the mode. This distribution is symmetric with respect to $(\lfloor \sqrt{n}\log n\rfloor+1)/2$ and is such that
 \begin{align*}
  p_n\left(\frac{\lfloor \sqrt{n}\log n\rfloor+1}{2}+k\right)&=p_n\left(\frac{\lfloor \sqrt{n}\log n\rfloor+1}{2}-k\right) =p_n\left(\frac{\lfloor \sqrt{n}\log n\rfloor+1}{2}\right)\prod_{i=1}^k\left(1-\frac{i-1/2}{n}\right),
 \end{align*}
 where $k\in\left\{1,\ldots,(\lfloor \sqrt{n}\log n\rfloor-1)/2\right\}$.

When $\lfloor \sqrt{n}\log n\rfloor$ is even, the distribution $p_n$ is bimodal with modes at $\lfloor \sqrt{n}\log n\rfloor/2$ and $\lfloor \sqrt{n}\log n\rfloor/2+1$. Using the same definitions as above for $a_{k,n}$ and $b_{k,n}$, we assume that
\begin{align*}
  &p_n(k+1)=a_{k,n}p_n(k), \quad  k\in\{\lfloor \sqrt{n}\log n\rfloor/2+1,\ldots,\lfloor \sqrt{n}\log n\rfloor-1\}, \cr
  &p_n(k-1)=a_{k-1,n}p_n(k), \quad  k\in\{2,\ldots,\lfloor \sqrt{n}\log n\rfloor/2\},
\end{align*}
 which implies that
 \begin{align}\label{eqn_prob_k_even}
  p_n\left(\frac{\lfloor \sqrt{n} \log n\rfloor}{2}+1+k\right)&=p_n\left(\frac{\lfloor \sqrt{n} \log n\rfloor}{2}-k\right)=p_n\left(\frac{\lfloor \sqrt{n} \log n\rfloor}{2}\right)\prod_{i=1}^k \left(1-\frac{i}{n}\right),
 \end{align}
 where $k\in\{1,\ldots,\lfloor \sqrt{n}\log n\rfloor/2-1\}$.

 The function $p_n(k)$ decreases at an exponential rate which is bounded below by $1/2$ when the distance between $k$ and the mode increases. The ratios $p_n(k+1)/p_n(k)$ and $p_n(k-1)/p_n(k)$ are indeed essentially bounded below by $1/2$. The case $\lfloor \sqrt{n}\log n\rfloor=5$ ($n=7$), for instance, is such that the ``best'' model has $n+(\lfloor \sqrt{n}\log n\rfloor+1)/2=10$ parameters, because the mode of $p_n$ is $K=3$ and the models have $n+K$ parameters, and the model with an additional parameter is less appropriate (so is the model with one less parameter), in the sense that $p_n(4)/p_n(3)=p_n(2)/p_n(3)=0.93$. The more parameters we add (or withdraw), the less appropriate the models are. Although $p_n$ is a pure construction, similar structures may arise in situations in which the models can be ranked by number of parameters and the posterior reflects the existence of a balance between overfitting (which involves a lot of parameters) and stability, in the spirit of Occam's razor. This is the case for instance in principal component regression. In this situation, the model indicator $K$ represents the number of components included in the model (i.e.\ Model $K=k$ is the model with the first $k$ components). It is easy to imagine situations where it would be optimal to include some, but not all, components. This would result in posterior distributions for the various models with structures as that described above (see, e.g., the real data analysis in \cite{gagnon2017robustPCR}). Note that, when $p_n$ has such a structure, it makes it easier for the RJ algorithm to explore the entire state space. Indeed the ratios $p_n(k+1)/p_n(k)$ and $p_n(k-1)/p_n(k)$ are never very small, which facilitates the transitions (see (\ref{eqn_prob_add}) and (\ref{eqn_prob_remove})).


 Finally, we add structure on the PMF $g$. Thanks to the Law of Large Numbers, the acceptance probability for updating the parameters (see (\ref{eqn_prob_update})) is (in some sense) asymptotically constant. It is however not possible to take advantage of the Law of Large Numbers to obtain the convergence of the acceptance probabilities for switching models. Adding structure on the PMF $g$, which represents the last required step towards the weak convergence result, will ease handling of these acceptance probabilities.
We consider that the function $g$ is as follows:
 \begin{align}\label{def_g}
  g(j):=\begin{cases}
        \tau \text{ if } j=1, \cr
        (1-\tau)A/(A+1) \text{ if } j=2, \cr
        (1-\tau)/(A+1) \text{ if } j=3,
       \end{cases}
 \end{align}
 where $0<\tau<1$ is a constant and $A:=2A^*$. The acceptance probability associated with the inclusion of an extra parameter (see (\ref{eqn_prob_add})) becomes the minimum between 1 and $f(U)/q(U)\times 1/A\times p_n(K+1)/p_n(K)$. By assumption, $2f/q\leq A$ and $p_n(K+1)/p_n(K)\leq 2$; therefore, this acceptance probability is simply $f(U)/q(U)\times 1/A\times p_n(K+1)/p_n(K)$. Furthermore, the acceptance probability associated with the withdrawal of the last parameter (see (\ref{eqn_prob_remove})) becomes the minimum between 1 and $q(X_{n+K}^{K})/f(X_{n+K}^{K})\times A\times p_n(K-1)/p_n(K)\geq 1$, which means that this type of movement is automatically accepted (whenever it is possible to withdraw a parameter, i.e.\ when $K>1$). Therefore, the probability to propose and accept switches from Model $K$ to Model $K-1$ is $(1-\tau)/(A+1)$. Proposition~\ref{prop_prob_acc_ajout} indicates that this is asymptotically equal to the ``average'' probability to propose and accept switches from Model $K$ to Model $K+1$. The probability that $K(m)$ increases by one is therefore the same as that of decreasing by one, globally (and asymptotically).
  \begin{prop}\label{prop_prob_acc_ajout}
  Consider the assumptions and the RJ algorithm described in this section. If we assume that $(K(0),\mathbf{X}^K(0))\sim\pi_n$, then for all $m\in\na$,
  $$
   \E\left[1 \wedge \frac{f(U(m+1))p_n(K(m)+1)g(3)}{q(U(m+1))p_n(K(m))g(2)}\right]\rightarrow \frac{1}{A} \text{ as } n\rightarrow\infty.
  $$
 \end{prop}
 \begin{proof} See Section~\ref{sec_proofs_props}. \end{proof}

\section{Towards Optimal Implementation of the RJ algorithm}\label{sec_result_weak}

In order to implement the RJ algorithm described in Section~\ref{sec_context_weak}, we have to specify the PDF $q$ and values for the constants $A$, $\tau$ and $\ell$. In Section~\ref{sec_problem_weak}, we present weak convergence results that are used in Section~\ref{sec_optimisation} to find asymptotically optimal (as $n\rightarrow \infty$) values for $\tau$ and $\ell$. This is what allows deriving the rules for tuning the RJ algorithm. 

\subsection{Weak Convergence Results}\label{sec_problem_weak}
 In order to study the asymptotic behaviour of the algorithm, we consider the following rescaled stochastic process:
 \begin{align}\label{def_Z}
  &\mathbf{Z}^n(t):= \left(\frac{K(\lfloor nt \rfloor )-\lfloor \sqrt{n}\log n \rfloor/2}{\sqrt{n}},\mathbf{X}^{K(\lfloor nt \rfloor )}(\lfloor nt \rfloor )\right),
 \end{align}
 where $t\geq0$. The continuous-time stochastic process $\{\mathbf{Z}^n(t): t\geq0\}$ is a sped up and modified version of $\{(K(m),\mathbf{X}^K(m)): m\in\na\}$. In any given iteration, the average jump size of the parameters
$$
 \mathbf{X}^{K(\lfloor nt \rfloor )}(\lfloor nt \rfloor ):=\left(X_1^{K(\lfloor nt \rfloor)}(\lfloor nt \rfloor),\ldots,X_{n+K(\lfloor nt \rfloor)}^{K(\lfloor nt \rfloor)}(\lfloor nt \rfloor)\right)
$$
 decreases with $n$ because the variance of the random walk is proportional to  $1/(n+{K(\lfloor nt \rfloor)}))$. The jump size of $\{K(\lfloor nt \rfloor )/\sqrt{n}: t\geq 0\}$ also decreases with $n$. In fact, each time it moves, its jump size is $1/\sqrt{n}$. The decreasing size of the jumps, combined with the acceleration of $\{\mathbf{Z}^n(t): t\geq0\}$, result in a continuous and non-trivial limiting process. We subtract $\lfloor \sqrt{n}\log n \rfloor/(2\sqrt{n})$ from $\{K(\lfloor nt \rfloor )/\sqrt{n}: t\geq 0\}$ in order to obtain a limiting process with components that take values on the real line. In particular, the asymptotic stationary distribution of this component of $\{\mathbf{Z}^n(t): t\geq0\}$, that we denote $\{Z_1^n(t):t\geq 0\}$, is a standard normal, as stated in Proposition~\ref{prop_z_1_to_normal}.

  \begin{prop}\label{prop_z_1_to_normal}
   Consider the assumptions described in Section~\ref{sec_context_weak}, the stochastic process $\{\mathbf{Z}^n(t): t\geq0\}$ defined in (\ref{def_Z}), and assume that $\mathbf{Z}^n(0)\sim\pi_n$. Then, as $n\rightarrow\infty$, $Z_1^n(t)$ converges in distribution towards a standard normal random variable, for all $t\geq 0$.
 \end{prop}

 \begin{proof} See Section~\ref{sec_proofs_props}. \end{proof}

 The main result is now stated.

 \begin{thm}\label{main_result}
  Consider the assumptions and the RJ algorithm described in Section~\ref{sec_context_weak}. Consider the stochastic process $\{\mathbf{Z}^n(t): t\geq0\}$ defined in (\ref{def_Z}) and assume that $\mathbf{Z}^n(0)\sim\pi_n$. Then, as $n\rightarrow\infty$, the first two components of $\{\mathbf{Z}^n(t):t\geq0\}$ converge weakly towards a bidimensional Langevin diffusion, i.e.\
  $$
  \{\mathbf{Z}_{1,2}^n(t): t\geq 0\}:=\{(Z_1^n(t),Z_2^n(t)): t\geq 0\}\Rightarrow \{\mathbf{Z}(t): t\geq 0\} \text{ as } n\rightarrow\infty,
  $$
  where the process $\{\mathbf{Z}(t): t\geq 0\}$ is comprised of two independent components such that $Z_1(0)\sim \mathcal{N}(0,1)$, $Z_2(0)\sim f$, and
  $$
  dZ_1(t)=-(1-\tau)/(A+1)\times Z_1(t)dt+\sqrt{2(1-\tau)/(A+1)}dB_1(t),
 $$
 $$
  dZ_2(t)=\tau \ell^2 \Phi(-\ell\sqrt{\Upsilon}/2)(\log f(Z_2(t)))' dt+\sqrt{2\tau \ell^2\Phi(-\ell\sqrt{\Upsilon}/2)}dB_2(t),
 $$
 with $\{B_1(t): t\geq0\}$ and $\{B_2(t): t\geq0\}$ being two independent Wiener processes, and $\Upsilon:=\E\left[(\log f(Z_2(0))')^2\right]$.
 \end{thm}

 \begin{proof} See Section~\ref{sec_proof_thm_1}. \end{proof}

 The notation ``$\Rightarrow$'' represents weak convergence (or convergence in distribution) of processes in the Skorokhod topology (for more details about this type of convergence, see Section 3 of \cite{ethier1986markov}).

 \begin{remark}
  The assumptions on the distribution of $K$ in Section~\ref{sec_context_weak} can be relaxed, under which case a result analogous to Theorem~\ref{main_result} holds. In particular, consider the transformed random variable $R^K$ defined as $R^K:=(K-\phi(n))/\sqrt{n}$ and suppose that the function $\phi(n)$ can be chosen such that
  \begin{align}\label{eqn_gen_cond}
   \frac{1}{p_n(K)}\frac{p_n(K+1)-p_n(K)}{1/\sqrt{n}}-(\log f_{Z_1}(R^K))'\rightarrow 0 \text{ \quad with probability 1 as $n\rightarrow\infty$},
  \end{align}
  where $f_{Z_1}$ is a strictly positive and smooth enough PDF. In that case, the first component of $\{\mathbf{Z}(t): t\geq 0\}$ is such that $Z_1(0)\sim f_{Z_1}$ and
 $$
  dZ_1(t)=(1-\tau)/(A+1) (\log f_{Z_1}(Z_1(t)))'dt+\sqrt{2(1-\tau)/(A+1)}dB_1(t).
 $$
 Given that the PMF $p_n$ of $K$ also is the PMF of $R^K$, and using the fact that $R^{K+1}-R^K=1/\sqrt{n}$, then the left term in (\ref{eqn_gen_cond}) can be seen as the discrete version of the derivative of $\log p_n$. One can thus analyse this term to identify the limiting process $\{\mathbf{Z}(t): t\geq 0\}$. Note that under the assumptions on $p_n$ stated in Section~\ref{sec_context_weak},
   \begin{align*}
   \frac{1}{p_n(K)}\frac{p_n(K+1)-p_n(K)}{1/\sqrt{n}}+R^K\rightarrow 0 \text{ \quad with probability 1 as $n\rightarrow\infty$},
  \end{align*}
 implying that $(\log f_{Z_1}(x))'=-x$, and therefore, that $f_{Z_1}=\mathcal{N}(0,1)$.
 \end{remark}

 \subsection{Optimisation}\label{sec_optimisation}

 The sample paths of $\{\mathbf{Z}(t): t\geq 0\}$ depend on $\tau, A,\ell$, $\Upsilon$ and $f$. In this section, we find values for $\ell$ and $\tau$ that are such that, for given $A$ and $f$ (and therefore $\Upsilon$), the state space exploration of this stochastic process is optimal. Optimising the asymptotic state space exploration of $\{(K(\lfloor nt \rfloor ),X_1^{K(\lfloor nt \rfloor)}(\lfloor nt \rfloor)): t\geq 0\}$ is sufficient to optimise that of $\{(K(\lfloor nt \rfloor ),\mathbf{X}^{K(\lfloor nt \rfloor)}(\lfloor nt \rfloor)): t\geq 0\}$. Indeed, in addition to optimising the exploration of the model space, we optimise the exploration of the first parameter space, and all parameters of all models share the same behaviour. During the theoretical optimisation, the constant $A$ is considered to be fixed because its value cannot be arbitrarily chosen. Indeed, it is tied to the ratio $f/q\leq A^*=A/2$. Note that the PDF $q$ only has an impact on the sample paths of $\{\mathbf{Z}(t): t\geq 0\}$ through the constant $A$.

 We first optimise the algorithm with respect to $\ell$. The stochastic process $\{Z_2(t): t\geq 0\}$ can be written as $Z_2(t)=V( 2\tau \ell^2 \Phi(-\ell\sqrt{\Upsilon}/2)\times t)$, where $\{V(t): t\geq 0\}$ is the following Langevin diffusion:
 $$
  dV(t)=(\log f(V(t)))'/2 \times dt+dB_2(t).
 $$
 The term $ 2\tau \ell^2 \Phi(-\ell\sqrt{\Upsilon}/2)$ that multiplies the time index of $\{V(t): t\geq 0\}$ to obtain $\{Z_2(t): t\geq 0\}$ is sometimes called the ``speed measure'' of $\{Z_2(t): t\geq 0\}$. Hereafter, when we discuss about to the ``speed'' of a stochastic process, we refer to this term. Viewed as a function of $\tau$ and $\ell$, the process $\{Z_2(t): t\geq 0\}$ that optimally explores its state space is thus the one with the largest speed. We can optimise the algorithm with respect to $\ell$ by maximising the speed of $\{Z_2(t): t\geq 0\}$ with respect to this variable, because the value of $\ell$ does not have an impact on the sample paths of $\{Z_1(t): t\geq 0\}$ (see the definition of $\{Z_1(t): t\geq 0\}$ in Theorem~\ref{main_result}). The function $2\tau \ell^2 \Phi(-\ell\sqrt{\Upsilon}/2)$ is maximised with respect to $\ell$ by $\ell=2.38/\sqrt{\Upsilon}$, as stated in Corollary 1.2 of \cite{roberts1997weak}. We therefore obtain the same optimal value as these authors. This is not surprising because the considered RJ algorithm updates the parameters in the same way as the RWM algorithm studied by these authors. Furthermore, the conditional distribution of the parameters given any model is essentially the same as their target distribution. In fact, the latter can be seen as a special case of the target that we consider, and their RMW algorithm can be seen as a special case of the considered RJ algorithm. It is thus interesting to observe that their result holds in our context.

 The optimisation with respect to $\ell$ tells us that the most efficient way to update the parameters is to set $\ell=2.38/\sqrt{\Upsilon}$. Therefore, the optimal variance for the random walk is $(2.38^2/(\Upsilon( n+K(m))))\mathcal{I}_{n+K(m)}$. As mentioned in \cite{roberts1997weak}, $\Upsilon$ can be seen as a measure of ``roughness'' of $f$. Consider for instance the case $f=\mathcal{N}(\mu,\sigma^2)$, which implies that $\Upsilon=1/\sigma^2$. In this case, we observe that small values for $\sigma$ correspond to ``rough'' $f$ functions, and vice versa. The ``rougher'' $f$ is, the smaller should be $\ell$. The result under normality also suggests that larger values for $\ell$ should be used when the tails are thicker.

 It could seem necessary to know $\Upsilon:=\E[(\log f(Z_2(0))')^2]$ in order to apply this optimal scaling result. Fortunately, the practical $0.234$ rule provided by \cite{roberts1997weak} can be used, as stated in Corollary~\ref{prop_conv_prob_acc}. This corollary is an adapted version of Corollary 1.2 of \cite{roberts1997weak}. Its proof is similar to the one given by these authors, and is thus omitted (it can nevertheless be found in \cite{gagnon2016thesis}).

 \begin{Corollary}\label{prop_conv_prob_acc}
  Consider the assumptions and the RJ algorithm described in Section~\ref{sec_context_weak}. Assuming that $(K(0),\mathbf{X}^K$ $(0))\sim\pi_n$, then, for all $m\in\na$,
  $$
   \E\left[1 \wedge \prod_{i=1}^{n+K(m)} \frac{f(Y_i^{K(m)}(m+1))}{f(X_i^{K(m)}(m))}\right]\rightarrow 2\Phi(-\ell\sqrt{\Upsilon}/2) \text{ as } n\rightarrow\infty.
  $$
  In addition, setting $\ell=2.38/\sqrt{\Upsilon}$ is equivalent to having $2\Phi(-\ell\sqrt{\Upsilon}/2)=2\Phi(-2.38/2)= 0.234$.
 \end{Corollary}

 Therefore, in order to reach optimal efficiency with respect to $\ell$, RJ users can monitor the acceptance rate of candidates $\mathbf{Y}^{K(m)}(m+1)$, where $\mathbf{Y}^{K(m)}(m+1)\sim \mathcal{N}(\mathbf{X}^{K(m)}(m),(\ell^2/(n+{K(m)}))\mathcal{I}_{n+K(m)})$, and tune the value of $\ell$ such that this rate is approximatively $0.234$. Note that this rate must be computed by considering only iterations in which it is proposed to update the parameters.


 We now optimise the algorithm with respect to $\tau$. We need a measure that takes into account the fact that an increase in the value of $\tau$ results in an increase in the speed of $\{Z_2(t): t\geq 0\}$, but also in a decrease in that of $\{Z_1(t): t\geq 0\}$, and vice versa. Essentially, when the value of $\tau$ is increased, more updates of the parameters (and therefore less model switches) are proposed. It would seem natural to consider the total speed of these two processes to optimise the algorithm with respect to $\tau$. The total speed is given by (using the optimal value for $\ell$)
 $$
  2\left[\tau(2.38^2/\Upsilon)\Phi(-2.38/2)+(1-\tau)/(A+1) \right].
 $$
 However, it is not a suitable measure because if, for instance $(2.38^2/\Upsilon)\Phi(-2.38/2)=(5.66/\Upsilon)$ $\Phi(-1.19)>1/(A+1)$, it would be proposed to choose the value of $\tau$ as close as possible to 1. In such a situation, there would be very few model switches, which would result in a slow exploration of the entire state space. We need a measure that penalises such a behaviour. This is achieved using integrated linear combinations of the autocorrelation functions (ACFs) of $\{Z_1(t): t\geq 0\}$ and $\{Z_2(t): t\geq 0\}$. Indeed, if for instance we set $\tau$ close to 1, $\{Z_2(t): t\geq 0\}$ would be a ``fast'' process with an ACF that decreases rapidly towards 0, while $\{Z_1(t): t\geq 0\}$ would be a ``slow'' process with an almost constant ACF around the value 1, which is undesirable. The sum of these two ACFs would decrease rapidly towards 1, thereafter remaining almost constant around this value. There should therefore exist a value of $\tau$ between 0 and 1 that induces two relatively ``fast'' processes, with a sum of ACFs that decreases relatively rapidly towards 0. We thus consider the integral of the sum of the ACFs of $\{Z_1(t): t\geq 0\}$ and $\{Z_2(t): t\geq 0\}$ to optimise the algorithm with respect to $\tau$:
\begin{align}\label{eqn_opt_mea}
  \int_0^\infty \left\{ \text{corr}[Z_1(t),Z_1(t+s)]+\text{corr}[Z_2(t),Z_2(t+s)]\right\}ds,\quad t\geq 0.
 \end{align}
 We drew inspiration from the effective sample size (see, e.g., Section 12.3.5 of \cite{robert2004monte}). The measure can be viewed as the sum of the (infinitesimally) integrated autocorrelation times of $\{Z_1(t): t\geq 0\}$ and $\{Z_2(t): t\geq 0\}$. It therefore represents a measure of the total ``inefficiency'' of these processes and the optimal value of $\tau$ is the one that minimises it.

We need to compute the ACFs of $\{Z_1(t): t\geq 0\}$ and $\{Z_2(t): t\geq 0\}$ in order to find the optimal value for $\tau$. The process $\{Z_1(t): t\geq 0\}$ satisfies the conditions of Theorem 2.1 in \cite{bibby2005diffusion}, implying that
 $$
  \text{corr}[Z_1(t),Z_1(t+s)]=\exp\left\{-(1-\tau)s/(A+1)\right\},\quad s,t\geq 0.
 $$
 The behaviour of $\{Z_2(t): t\geq 0\}$ depends on $f$ and consequently its ACF cannot be computed in all generality. A particular situation is now studied in order to obtain general information about the optimal values of $\tau$. We consider the case $f=\mathcal{N}(\mu,\sigma^2)$, $\mu\in\re,\sigma>0$, which allows to obtain an explicit solution to our optimisation problem. Using the optimal value for $\ell$, we have that
 \begin{align*}
  dZ_2(t)&=-\tau 5.66 \sigma^2 \Phi\left(-1.19\right)\times (Z_2(t)-\mu)/\sigma^2 dt+\sqrt{2\tau 5.66 \sigma^2\Phi\left(-1.19\right)}dB_2(t).
 \end{align*}
 This process also satisfies the conditions of Theorem 2.1 in \cite{bibby2005diffusion}, implying that
 $$
  \text{corr}[Z_2(t),Z_2(t+s)]=\exp\left\{-\tau 5.66 \Phi\left(-1.19\right) s\right\},\quad s,t\geq 0.
 $$
 Thus, when $f=\mathcal{N}(\mu,\sigma^2)$ and $\ell$ is set to its optimal value, the optimal value for $\tau$ is
  \begin{align}\label{eqn_tau_opt}
  & \frac{\sqrt{5.66 \Phi\left(-1.19\right)(A+1)}-1}{5.66 \Phi\left(-1.19\right)(A+1)-1}.
 \end{align}
 The constant $A$ clearly has an impact on the optimal value for $\tau$, as intuitively explained in the introduction. The optimal value essentially decreases as $1/\sqrt{A}$ (see Figure~\ref{fig_int_tauopt_A}). 
 When $A=2$, the optimal value for $\tau$ is $0.415$. This situation corresponds to $f=q$, and therefore to the best choice of distribution $q$. When $A=5$, the optimal value for $\tau$ is $0.334$, and when $A=25$, it is $0.194$. The constant $A$ therefore has an indirect impact on the sample paths of $\{Z_2(t): t\geq 0\}$ through the optimal choice for $\tau$. In other words, the larger is $A$, the higher should be the probability to propose model switches, which implies less parameter updates.
 Notice from Figure~\ref{fig_int_tauopt_A} that the integrals of the sum of the ACFs do not vary much when $\tau$ is between $0.2$ and $0.6$ for the cases $A=2$ and $A=5$. As will be observed in the numerical experiment in Section~\ref{sec_simulation_weak}, this is also the case when $A=25$, but $n$ is not too large. This suggests that, when either $A$ or $n$ is not too large, setting $\tau$ to any value in this interval leads to algorithms having similar efficiencies. As $n$ increases, users should however narrow down to the optimal values for $\tau$ in the case where $A$ is moderate or large.

 \begin{figure}[h]
 \centering
    \includegraphics[width=15cm]{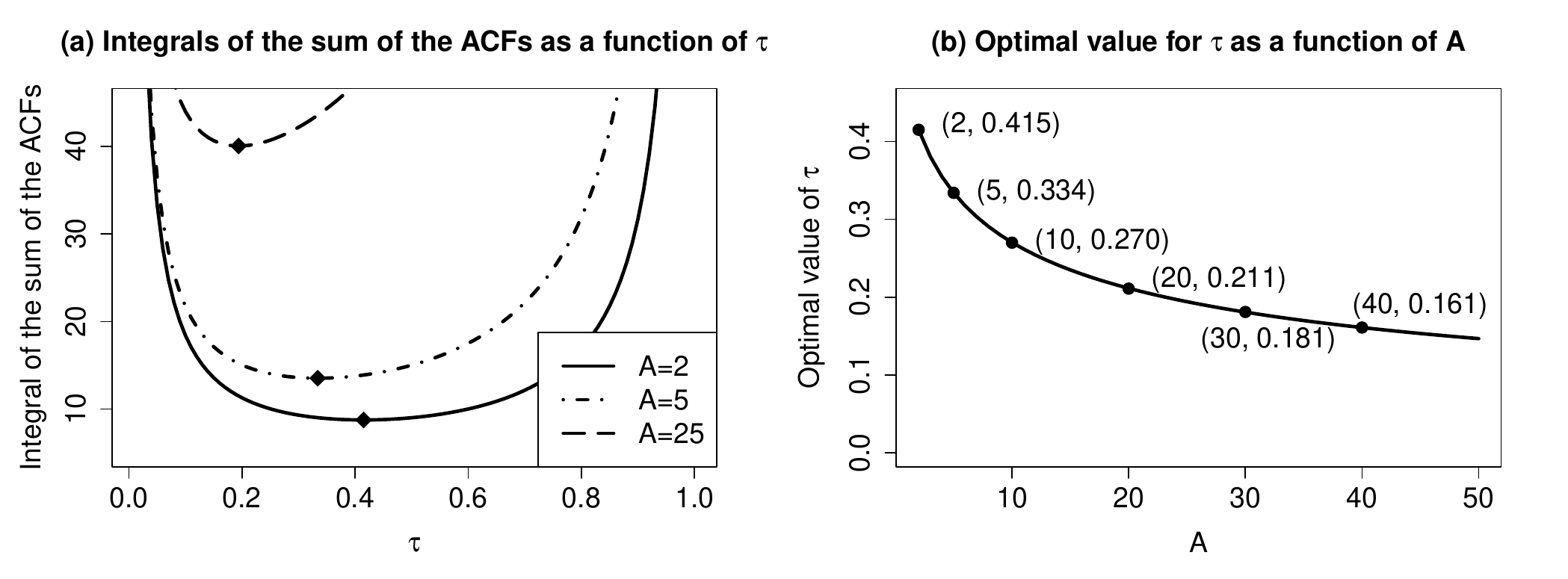}
    \vspace{-5mm}
  \caption{\small (a) Integrals of the sum of the ACFs as a function of $\tau$ for $A=2,5,25$ (the diamonds represent the optimal value for $\tau$), (b) Optimal value for $\tau$ as a function of $A$; for both graphs, $f=\mathcal{N}(\mu,\sigma^2)$ and $\ell$ is set to its optimal value}\label{fig_int_tauopt_A}
 \end{figure}

 Different distances between $f$ and $q$ lead to different upper bounds on $f/q$ (which are given by $A/2$), and also different acceptance probabilities for switching models. In practice, the former are much more difficult to evaluate than the latter. In addition, we believe that in general the behaviours of the acceptance probabilities are not summarised by the upper bound on $f/q$. This is why we propose to tune the distribution $g$ according to the acceptance rates.

\section{Practical Considerations}

 \subsection{Optimal Implementation and Generalisation}\label{sec_optimal_impl}

  Now that the framework is properly defined and the results presented, we are able to state the rules for tuning the RJ algorithm presented in the introduction more precisely. We first recommend to perform some trial runs. They shall be used to identify the value of $\ell$ that is such that the proposed parameter updates are accepted with a rate of $0.234$. At the same time, record the number of times that model switches are proposed and accepted, i.e.\ the number of times that $K$ moves, and compute the rate of these events. If this rate is not too low or the dimensions of the models are not too large, setting $\tau$ to any value in the interval $(0.2, 0.6)$ should be suitable. We recommend to compute its optimal value as a reference. As $n$ increases, users should increasingly consider the optimal value. In contexts as ours, considering that the rate is, theoretically (and asymptotically), $2(1-\tau)/(A+1)$ (see Section~\ref{sec_context_weak}), the optimal value can be determined through (\ref{eqn_tau_opt}). We believe this result leads to appropriate values in more general settings. We therefore propose to consider the following rule in these: divide the rate by $1-\tau$ (denote the result $r$), and compute the ``optimal'' value for $\tau$ using $(\sqrt{1/r}-1)/(1/r-1)$. This rule has been derived from (\ref{eqn_tau_opt}), in which $5.66\Phi(-1.19)=0.66$ has been replaced by $1/2$. The acceptance rate of each different type of model switches should be recorded as well. As mentioned in Section~\ref{sec_context_weak}, it is preferable to avoid very small rates. The information gathered about $f$ can be used to improve the design of the proposal distribution $q$, which will lead to higher rates. The process is in this sense iterative.


As mentioned in the introduction, the optimal scaling result of \cite{roberts1997weak} is known to be relatively robust, in the sense that it holds under weaker assumptions (see, e.g., \cite{roberts2001optimal} and \cite{bedard2007weak}). Believing that our theoretical results are also robust, we conjecture that they are valid when
\begin{align}\label{eqn_post_gen}
 \pi_n(k,\mathbf{x}^k)=p_n(k)\prod_{i=1}^{n+k} f_i(x_i^k),
\end{align}
where $f_i(x_i^k):=(1/\sigma_i) f((x_i^k-\mu_i)/\sigma_i)$, $\mu_i\in\re$ and $\sigma_i>0$ are constants, and $f$ and $p_n$ satisfy the assumptions described in Section~\ref{sec_context_weak}. In this case, the sampler described in Section~\ref{sec_context_weak} is applied, the only difference being that a proposal distribution $q_{K+1}$ is used to add a parameter to switch from Model $K$ to Model $K+1$, in order to accommodate for the different functions $f_i$. The corresponding assumption on $f/q$ is $f_i/q_i\leq A_i^*\leq A_n^*$ for all $i\in\{n+1,\ldots,n+\lfloor \sqrt{n}\log n\rfloor\}$ with $A_n^*:=\max_{i} A_i^*$, and $A_n^*\leq A^*$ for all $n$. A similar procedure as that described previously for optimising the algorithm would be applied. We thus expect the results in this paper to be applicable to models in which the parameters are independent but not identically distributed. In particular, they turn out to be useful in designing the RJ algorithm when the posterior distribution arises from a robust principal component regression, as shown in \cite{gagnon2017robustPCR}. This is explained by the fact that the posterior has a structure similar to that in  (\ref{eqn_post_gen}). The authors are consequently able to design an efficient sampler that allows identification of the relevant components and estimation of the parameters for prediction purpose. The generalised version of the algorithm (described above) is applied. It is observed that the optimal value for $\ell$ corresponds to an acceptance rate close to $0.234$ and the optimal value for $\tau$ is close to $0.6$, which makes sense considering that the dimensions are small in their situation. Note that the generalisation of our results to contexts like this one is not trivial.

\subsection{A Numerical Experiment}\label{sec_simulation_weak}

 Samples produced by the RJ algorithm provide information about the joint posterior distribution of the models and their parameters. As a result, users can select models (usually those with the highest frequencies in the sample) and estimate their parameters (using sample means and intervals for instance). Ideally, the selected models, along with the parameter estimates, would be the same as if the ``true'' posterior distribution had been used. In this section, we show what is the impact of the design of the sampler on the quality of the approximations of the posterior model probabilities and parameter estimates. We consider the framework defined in this paper, under which the optimal design is known. This will also allow to understand how the theoretical results presented in Section~\ref{sec_result_weak} translate in practice.

Implementing the RJ algorithm described in Section~\ref{sec_context_weak} comes down to specifying the PDF $q$ and the values of the constants $A,\tau$ and $\ell$. In Section~\ref{sec_result_weak}, it has been shown that the constants $\tau$ and $A$ have an impact on the estimation of the whole joint posterior distribution of $(K, \mathbf{X}^K)$. The constant $\ell$ has an impact on the $\mathbf{X}^K$ part only and this has been thoroughly studied by \cite{roberts2001optimal}. In Section~\ref{sec_result_weak}, it has also been explained that, asymptotically, the PDF $q$ only has an impact through the constant $A$. We can therefore reproduce any distance between $f$ and $q$ by only changing $A$. In this section, we therefore focus on showing the impact of the constants $\tau$ and $A$ on the approximations. In other words, we focus on showing how different probabilities of proposing parameter updates (and therefore model switches) and different distances between $f$ and $q$ have an impact on the approximations. We consider $f=\mathcal{N}(\mu,\sigma^2)$, $\mu\in\re,\sigma>0$, and we set $\ell=2.38\sigma$ (the optimal value) and $q=f$. We evaluate the performance of the algorithm for every $\tau\in( 0, 1)$, in the cases where $A=2,5,25$. For a given $A$, the optimal value for $\tau$ can thus be determined using (\ref{eqn_tau_opt}).

For fixed $\tau$ and $A$, we evaluate the performance of the RJ algorithm using mean absolute deviations (MADs) around quantities that are usually of interest for users: the posterior mode of $K$ (denoted by $k^*$), and the posterior mean and standard deviation of $X_i^K$, $i\in\{1,\ldots,n+K\}$ (we consider the first parameter $X_1^{K}$ for the experiment). For a given sample produced by the RJ algorithm, $k^*$ is estimated by $\widehat{k^*}$, the mode of the sample associated with the random variable $K$, and $\mu$ and $\sigma$ are estimated by $\hat{\mu}$ and $\hat{\sigma}$, which are respectively the mean and standard deviation of the sample associated with the random variable $X_i^K$. Adequate samples lead to accurate estimates, thus resulting in small absolute deviations. For fixed $\tau$ and $A$, we approximate the MADs by independently running $N$ times the RJ algorithm and by computing $\sum_{i=1}^N |\widehat{k^*_{i}}-k^*|/N$, $\sum_{i=1}^N\abs{\hat{\mu}_{i}-\mu}/N$ and $\sum_{i=1}^N\abs{\hat{\sigma}_{i}-\sigma}/N$, where $\widehat{k^*_{i}}$, $ \hat{\mu}_{i}$ and $\hat{\sigma}_{i}$ are respectively the mode, mean and standard deviation based on the sample produced by the $i$th run.  We also compute a global measure that we expect to have the same behaviour as that used in Section~\ref{sec_result_weak} (and given in~(\ref{eqn_opt_mea})) to optimise the algorithm with respect to $\tau$. This global measure is a linear combination of a standardised version of the MADs:
\small$$
 \frac{\frac{1}{N}\sum_{i=1}^N |\widehat{k^*_{i}}-k^*|}{\sqrt{\frac{1}{N-1}\sum_{i=1}^N (\widehat{k^*_{i}}-k^*)^2}}+\frac{1}{2}\left(\frac{\frac{1}{N}\sum_{i=1}^N\abs{\hat{\mu}_{i}-\mu}}{\sqrt{\frac{1}{N-1}\sum_{i=1}^N(\hat{\mu}_{i}-\mu)^2}}+\frac{\frac{1}{N}\sum_{i=1}^N\abs{\hat{\sigma}_{i}-\sigma}}{\sqrt{\frac{1}{N-1}\sum_{i=1}^N(\hat{\sigma}_{i}-\sigma)^2}}\right).
$$\normalsize
In this experiment, $N=\,$1,000, $\mu=0$, $\sigma=1$, and each sample is of size 100,000. The results are presented in Figure~\ref{fig_performance}.

As expected, the quality of the sample associated with $K$ decreases as the value of $\tau$ increases due to fewer model switches. An increase in $\tau$ has the opposite effect regarding the quality of the sample associated with $\mathbf{X}^K$. The vertical lines represent the optimal values for $\tau$, which are $0.415$, $0.334$ and $0.194$, when $A=2, 5, 25$, respectively. These values are optimal in the sense that they allow to reach the appropriate balance between adequate samples for $K$ (but inadequate ones for $\mathbf{X}^K$) and adequate samples for $\mathbf{X}^K$ (but inadequate ones for $K$).

Figure~\ref{fig_performance} also helps illustrate that the smaller is $A$, the better it is,
an aspect that has been theoretically justified in Section~\ref{sec_result_weak}. Indeed, the value of $A$ has a direct impact on the quality of the samples for $K$ (it decreases as the value of $A$ increases), and it has an indirect impact on the quality of those for $\mathbf{X}^K$ through the optimal value for $\tau$.

We finally note that, as $n\rightarrow\infty$, the curves defined by the global measure as a function of $\tau$ look more and more like those in Figure~\ref{fig_int_tauopt_A} (a). For moderate values of $n$, the curves defined by the global measure are however almost flat between 0.2 and 0.6.

\section{Conclusion}\label{sec_conclusion_weak}

In this paper, we have provided guidelines for implementing an efficient RJ algorithm. They rely on theoretical results that hold when the algorithm is applied to sample from a target distribution $\pi_n$ that satisfies the assumptions provided in Section~\ref{sec_context_weak}. Essentially, the target dis\-tribution $\pi_n$ must be a product of the PMF $p_n$ (the distribution of the model indicator $K$) and a $(n+K)$-product of PDFs $f$. Being aware that this sampling context is simple, our goal was to open new research directions in optimal tuning of RJ algorithms and provide rules that we believe are robust. They seem valid for tuning the sampler when the posterior arise from robust principal component regression, as explained in Section~\ref{sec_optimal_impl}. We conjectured that they are and that, more generally, our results hold when $\pi_n$ is  comprised of a product of different functions $f_i$.

\section{Proof of Theorem \ref{main_result}}\label{sec_proof_thm_1}

 This section is dedicated to the demonstration of the main result of this paper, the weak convergence $\{\mathbf{Z}_{1,2}^n(t):t\geq 0\}\Rightarrow \{\mathbf{Z}(t):t\geq 0\}$ in the Skorokhod topology as $n\rightarrow\infty$ (the stochastic processes $\{\mathbf{Z}_{1,2}^n(t):t\geq0\}$ and $\{\mathbf{Z}(t):t\geq0\}$ have been defined in Theorem~\ref{main_result}). Thus consider the sampling context described in Section~\ref{sec_context_weak}.

 In order to prove the result, we demonstrate the convergence of the finite-dimensional distributions of $\{\mathbf{Z}_{1,2}^n(t): t\geq0\}$ to those of $\{\mathbf{Z}(t): t\geq0\}$. To achieve this, we verify condition (c) of \makebox[\textwidth][s]{Theorem 8.2 from Chapter 4 of \cite{ethier1986markov}. The weak convergence then follows from}

   \begin{figure}[H]
  \centering
  $\begin{array}{ccc}
    \multicolumn{3}{c}{\hskip-0.5cm \includegraphics[width=14cm]{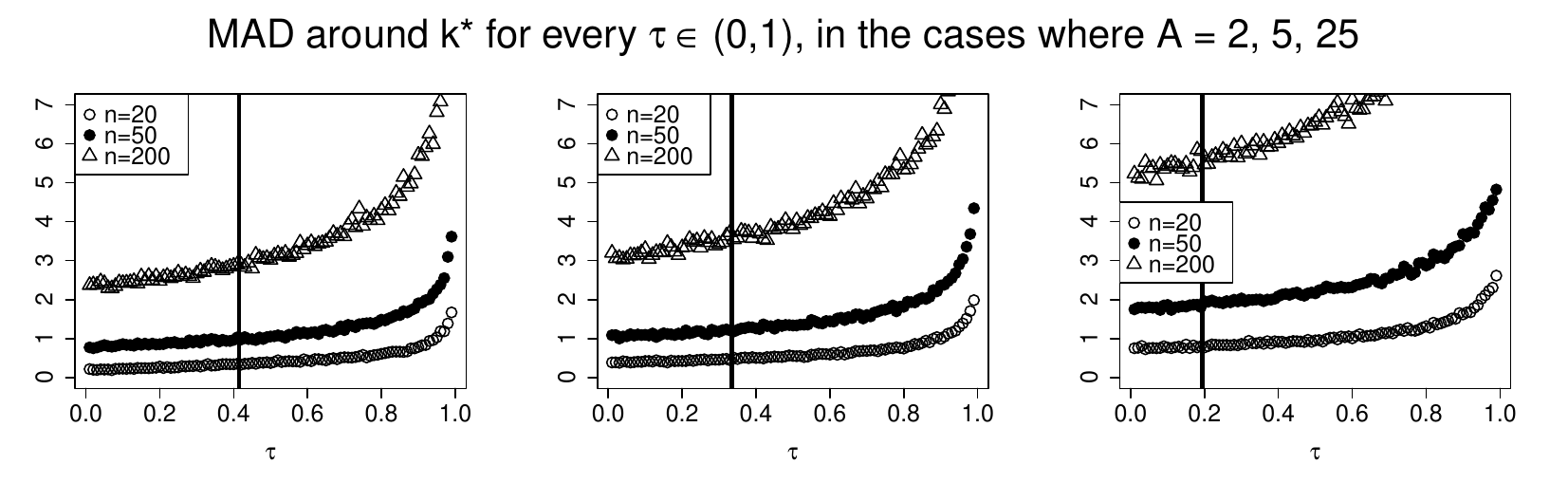}} \cr
    \multicolumn{3}{c}{\hskip-0.5cm \includegraphics[width=14cm]{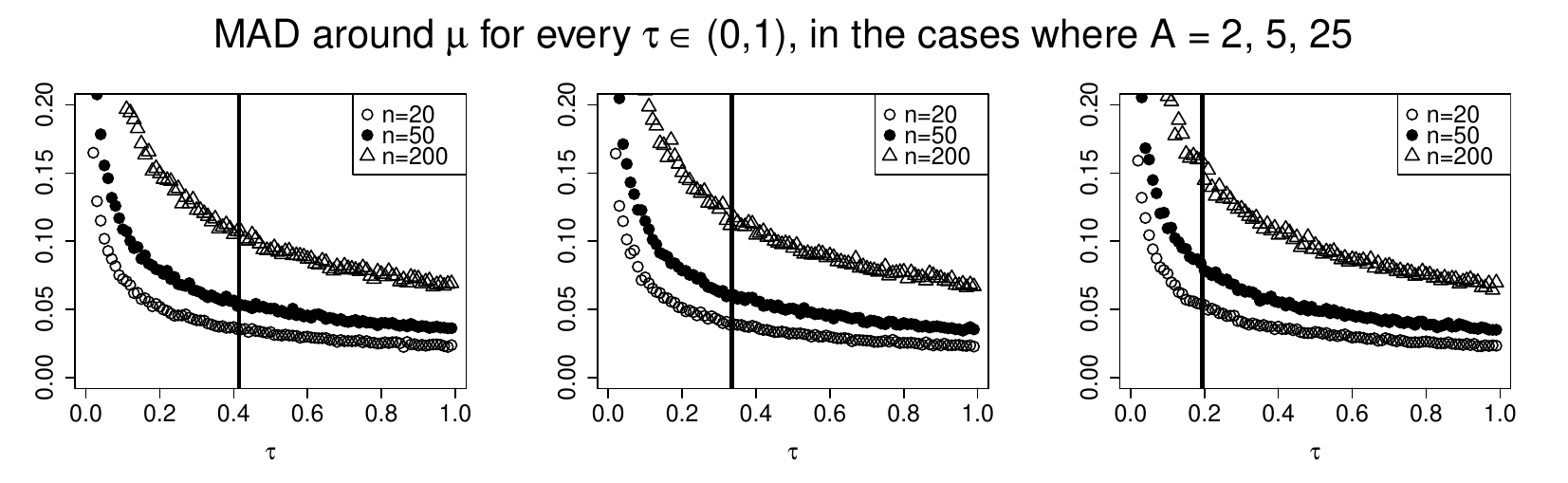}} \cr
   \multicolumn{3}{c}{\hskip-0.5cm \includegraphics[width=14cm]{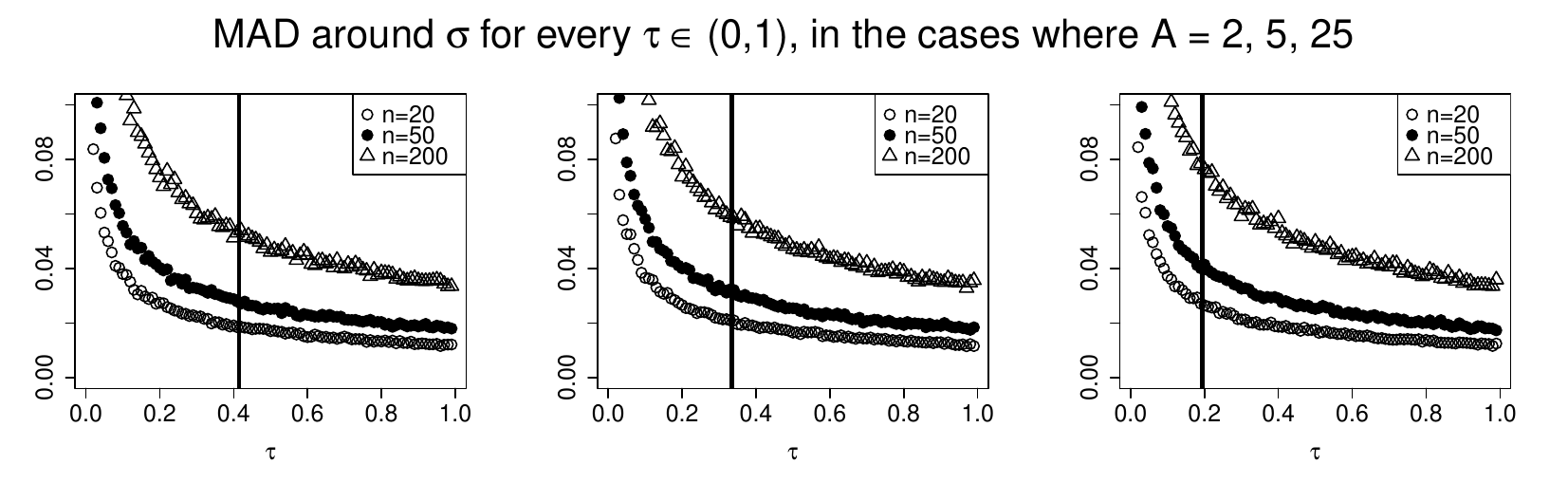}} \cr
  \multicolumn{3}{c}{\hskip-0.5cm \includegraphics[width=14cm]{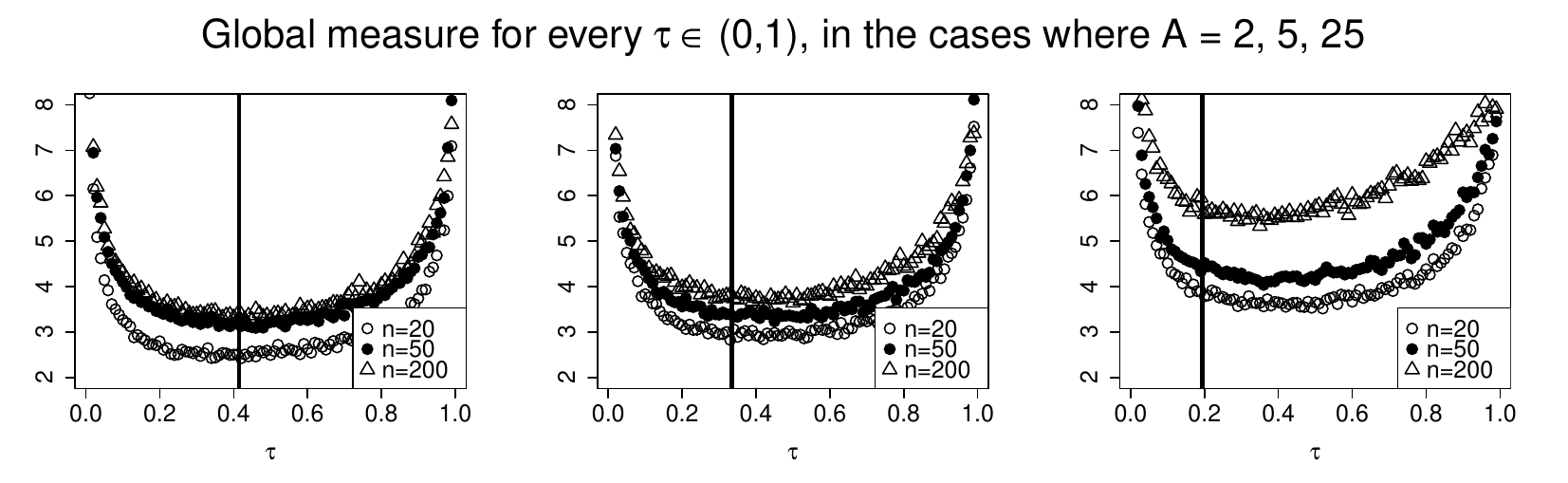}} \cr
    \hspace{1cm} \textbf{(a) \textit{A = 2}} & \hspace{2.5cm} \textbf{(b) \textit{A = 5}} & \hspace{1cm}\textbf{(c) \textit{A = 25}} \cr
  \end{array}$
  \caption{\small MADs around $k^*$, $\mu$ and $\sigma$, and the global measure, for every $\tau\in( 0, 1)$, in the cases where $A=2,5,25$ (the vertical lines represent the optimal values for $\tau$, which are $0.415$, $0.334$ and $0.194$, when $A=2, 5, 25$, respectively)}\label{fig_performance}
 \end{figure}

 \noindent Corollary 8.6 of Chapter 4 of \cite{ethier1986markov}. The remaining conditions of Theorem 8.2 and the conditions specified in Corollary 8.6 are either straightforward or easily derived from the proof given in this section. They are nevertheless explicitly verified in \cite{gagnon2016thesis}.

 The proof of the convergence of the finite-dimensional distributions relies on the convergence of (what we call) the ``pseudo-generator'', a quantity that we now introduce. The proof follows in Section~\ref{sec_proof_finite_dim}.

 \subsection{Pseudo-Generator}\label{sec_generator}

In this section, we introduce a quantity that we call the ``pseudo-generator" of $\{\mathbf{Z}_{1,2}^n(t): t\geq 0\}$ due to its similarity with the infinitesimal generator of stochastic processes. It is defined as follows:
$$
 \varphi_{n}(t):=n\E[h(\mathbf{Z}_{1,2}^n(t+1/n))-h(\mathbf{Z}_{1,2}^n(t))\mid \F^{\mathbf{Z}^n}(t)],
$$
where  $h\in\mathcal{C}_c^\infty(\re^2)$, the space of infinitely differentiable functions on $\re^2$ with compact support. Theorem 2.5 from Chapter 8 of \cite{ethier1986markov} allows us to restrict our attention to this set of functions when studying the limiting behaviour of the pseudo-generator (see \cite{gagnon2016thesis} for more details).

Let
 $$
  R^K(\lfloor nt \rfloor):=\frac{K(\lfloor nt \rfloor )-\lfloor \sqrt{n}\log n \rfloor/2}{\sqrt{n}}=Z_1^n(t).
 $$
The pseudo-generator $\varphi_{n}(t)$ can be decomposed into three parts, each associated with a specific type of movement, as follows:
 $$
  \varphi_{n}(t)=\varphi_{1,n}(t)+\varphi_{2,n}(t)+\varphi_{3,n}(t),
 $$
 where $\varphi_{1,n}(t)$ is associated with the update of the parameters, i.e.\
 \begin{align*}
  \varphi_{1,n}(t)&:=n\tau\E\left[\left(h(R^K,Y_1^K)-h(R^K,X_1^K)\right)\left(1 \wedge \frac{\prod_{i=1}^{n+K} f(Y_i^{K})}{\prod_{i=1}^{n+K} f(X_i^{K})}\right)\mid R^K,\mathbf{X}^K\right],
 \end{align*}
  $\varphi_{2,n}(t)$ is associated with the inclusion of an extra parameter, i.e.\
 \begin{align}\label{def_phi2}
  \varphi_{2,n}(t)&:=\frac{n(1-\tau)A}{A+1}\E\left[\left(h(R^{K+1},X_1^K)-h(R^K,X_1^K)\right)\left(1 \wedge \frac{f(U)p_n(K+1) }{q(U)p_n(K)A}\right)\mid R^K,\mathbf{X}^K\right],
 \end{align}
 and $\varphi_{3,n}(t)$ is associated with the withdrawal of the last parameter, i.e.\
 \begin{align}\label{def_phi3}
  \varphi_{3,n}(t)&:=\frac{n(1-\tau)}{A+1}\E\left[\left(h(R^{K-1},X_1^K)-h(R^K,X_1^K)\right)\left(1 \wedge \frac{q(X_{n+K}^{K})p_n(K-1)A }{f(X_{n+K}^{K})p_n(K)}\right)\mid R^K,\mathbf{X}^K\right].
 \end{align}
 Note that the Markov process $\{(R^K(m),\mathbf{X}^{K(m)}(m)): m\in\na\}$ is time-homo\-ge\-neous, and consequently, the time index has been omitted to simplify the notation. Also note that, when there is an update of the parameters, only the parameters $\mathbf{X}^K$ move (the model indicator remains the same). When an extra parameter is included or the last parameter withdrawn, only the model indicator changes, as a switch from Model $K$ to Model $K+1$ or $K-1$ is made.

 \subsection{Proof of the Convergence of the Finite-Dimensional Distributions}\label{sec_proof_finite_dim}

 Condition (c) of Theorem 8.2 from Chapter 4 of \cite{ethier1986markov} essentially reduces to the following convergence:
 \begin{align*}
  \E\left[\abs{\varphi_n(t)-Gh(\mathbf{Z}_{1,2}^n(t))}\right]\rightarrow 0 \text{ as } n\rightarrow \infty,
 \end{align*}
 where $G$ is the generator of a diffusion with $G=G_1+G_2$ and
 \begin{align*}
  G_1h(\mathbf{Z}_{1,2}^n(t))&=\tau \ell^2 \Phi\left(-\frac{\ell\sqrt{\Upsilon}}{2}\right)(\log f(Z_2^n(t)))' h_y(\mathbf{Z}_{1,2}^n(t))+\tau \ell^2\Phi\left(-\frac{\ell\sqrt{\Upsilon}}{2}\right)h_{yy}(\mathbf{Z}_{1,2}^n(t)), \cr
  G_2h(\mathbf{Z}_{1,2}^n(t))&=\frac{1-\tau}{A+1}\times -Z_1^n(t)h_{x}(\mathbf{Z}_{1,2}^n(t))+\frac{1-\tau}{A+1}h_{xx}(\mathbf{Z}_{1,2}^n(t)).
 \end{align*}
 The function $h$ above is the same function $h$ involved in the definition of the random variable $\varphi_{n}(t)$ (given in Section~\ref{sec_generator}). In other words, the convergence has to be proved for an arbitrary function $h\in\mathcal{C}_c^\infty(\re^2)$. The functions $h_x$ and $h_{xx}$ respectively represent the first and second derivatives of $h$ with respect to its first argument. Analogously, the functions $h_y$ and $h_{yy}$ respectively represent the first and second derivatives of $h$ with respect to its second argument. Note that it exists a positive constant $M$ such that $h$ and all its derivatives are bounded in absolute value by this constant.

 Using the triangle inequality, we have
 \begin{align*}
  \E\left[\abs{\varphi_n(t)-Gh(\mathbf{Z}_{1,2}^n(t))}\right]&\leq \E\left[\abs{\varphi_{1,n}(t)-G_1h(\mathbf{Z}_{1,2}^n(t))}\right]+\E\left[\abs{\varphi_{2,n}(t)+\varphi_{3,n}(t)-G_2h(\mathbf{Z}_{1,2}^n(t))}\right].
 \end{align*}
 In this paper, we show that the second term on the right-hand side (RHS) converges towards 0 as $n\rightarrow\infty$. The proof that the first term converges towards 0 is similar to that of Theorem~1.1 of \cite{roberts1997weak}, and is thus omitted (it can nevertheless be found in \cite{gagnon2016thesis}).

 The key here is the use of Taylor expansions in order to obtain derivatives of $h$ as in generators of diffusions.

 We first analyse $\varphi_{2,n}(t)$ as defined in (\ref{def_phi2}). As explained in Section~\ref{sec_context_weak}, $0\leq p_n(K+1)/p_n(K)\leq 2$, and therefore,
 $$
  \frac{f(U)p_n(K+1) }{q(U)p_n(K)A}\leq \frac{2f(U)}{q(U)A}\leq 1.
 $$
 Consequently, since $h(R^{K+1},X_1^K)=h(R^{K}+1/\sqrt{n},X_1^K)$,
 \begin{align*}
  \varphi_{2,n}(t)&
  =\frac{n(1-\tau)}{A+1}\left(h(R^K+1/\sqrt{n},X_1^K)-h(R^K,X_1^K)\right)\frac{p_n(K+1)}{p_n(K)}\E\left[\frac{f(U) }{q(U)}\mid  R^K,\mathbf{X}^K\right] \cr
  &=\frac{n(1-\tau)}{A+1}\left(h(R^K+1/\sqrt{n},X_1^K)-h(R^K,X_1^K)\right)\frac{p_n(K+1)}{p_n(K)}.
 \end{align*}
 In the last equality, we use the fact that $U$ is independent of $(K,\mathbf{X}^K)$, and therefore,
 \begin{align*}
   \E\left[ \frac{f(U) }{q(U)}\mid  R^K,\mathbf{X}^K\right]=\E\left[ \frac{f(U) }{q(U)} \right]=\int_{-\infty}^\infty \frac{f(u)}{q(u)}q(u)du=1.
 \end{align*}
 Note that $\varphi_{2,n}(t)=0$ when $K=\lfloor \sqrt{n}\log n\rfloor$ since $p_n( \lfloor \sqrt{n}\log n\rfloor+1)=0$.

  We now study $\varphi_{3,n}(t)$ as defined in (\ref{def_phi3}). As explained in Section~\ref{sec_context_weak}, when $2\leq K\leq \lfloor\sqrt{n}\log n\rfloor$ we have that $p_n(K-1)/p_n(K)\geq 1/2$. Therefore, when $2\leq K\leq \lfloor\sqrt{n}\log n\rfloor$,
 $$
  \frac{q(X_{n+K}^{K})p_n(K-1)A }{f(X_{n+K}^{K})p_n(K)}\geq \frac{q(X_{n+K}^{K})A }{2f(X_{n+K}^{K})}\geq 1.
 $$
 This means that the acceptance probability of withdrawing the last parameter is 1, when it is possible to withdraw a parameter. Consequently, since $h(R^{K-1},X_1^K)=h(R^{K}-1/\sqrt{n},X_1^K)$,
 \begin{align*}
 &\varphi_{3,n}(t)=\frac{n(1-\tau)}{A+1}\left(h(R^{K}-1/\sqrt{n},X_1^K)-h(R^{K},X_1^K)\right)\ind(2\leq K\leq \lfloor\sqrt{n}\log n\rfloor).
 \end{align*}
 Note that $\varphi_{3,n}(t)=0$ when $K=1$ since $p_n(0)=0$.

 By using Taylor expansions of $h$ around $R^K$, we obtain that
 \begin{align*}
  h(R^K+1/\sqrt{n},X_1^K)-h(R^K,X_1^K)&=\frac{1}{\sqrt{n}}h_x(R^K,X_1^K)+\frac{1}{2n}h_{xx}(R^K,X_1^K)+\frac{1}{6n^{3/2}}h_{xxx}(W,X_1^K), \cr
  h(R^K-1/\sqrt{n},X_1^K)-h(R^K,X_1^K)&=-\frac{1}{\sqrt{n}}h_x(R^K,X_1^K)+\frac{1}{2n}h_{xx}(R^K,X_1^K)-\frac{1}{6n^{3/2}}h_{xxx}(T,X_1^K),
 \end{align*}
 where $W$ belongs to $(R^K,R^K+1/\sqrt{n})$, $T$ belongs to $(R^K-1/\sqrt{n},R^K)$, and $h_{xxx}$ represents the third derivative of $h$ with respect to its first argument.

 Therefore,
 \normalsize\begin{align}\label{phi2_phi_3_gen}
  &\varphi_{2,n}(t)+\varphi_{3,n}(t)-G_2h(\mathbf{Z}_{1,2}^n(t))=\ind(2\leq K\leq \lfloor \sqrt{n}\log n\rfloor-1) \cr
  &\qquad\times \frac{1-\tau}{A+1}h_x(R^K,X_1^K) \left[\sqrt{n}\left( \frac{p_n(K+1) }{p_n(K)}-1\right)+R^K\right] \cr
  &+\ind(K=1)\frac{1-\tau}{A+1}h_x(R^K,X_1^K)\left(\sqrt{n}\times\frac{p_n(K+1) }{p_n(K)}+R^K\right) \cr
  &-\ind(K=\lfloor\sqrt{n}\log n\rfloor)\frac{1-\tau}{A+1}h_x(R^K,X_1^K)\left(\sqrt{n}-R^K\right) \cr
  &+\ind(2\leq K\leq \lfloor \sqrt{n}\log n\rfloor-1)\frac{1-\tau}{2(A+1)}h_{xx}(R^K,X_1^K)\left( \frac{p_n(K+1) }{p_n(K)}-1\right) \cr
  &+\ind(K=1)\frac{1-\tau}{2(A+1)}h_{xx}(R^K,X_1^K)\left(\frac{p_n(K+1) }{p_n(K)}-2\right) \cr
  &-\ind(K=\lfloor\sqrt{n}\log n\rfloor)\frac{1-\tau}{2(A+1)}h_{xx}(R^K,X_1^K) \cr
  &+\frac{1-\tau}{6\sqrt{n}(A+1)}h_{xxx}(W,X_1^K) \frac{p_n(K+1) }{p_n(K)}\ind(1\leq K\leq \lfloor\sqrt{n}\log n\rfloor-1) \cr
  &-\frac{1-\tau}{6\sqrt{n}(A+1)}h_{xxx}(T,X_1^K)\ind(2\leq K\leq \lfloor\sqrt{n}\log n\rfloor).
 \end{align} \normalsize
 We now show that the expectation of the absolute value of each term on the RHS in (\ref{phi2_phi_3_gen}) converges towards 0 as $n\rightarrow\infty$. Consequently, using the triangle inequality we will obtain
 $$
  \E\left[\abs{\varphi_{2,n}(t)+\varphi_{3,n}(t)-G_2h(\mathbf{Z}_{1,2}^n(t))}\right]\rightarrow 0 \text{ as } n\rightarrow\infty.
 $$
   We start with the last terms in (\ref{phi2_phi_3_gen}) and make our way up. It is clear that the expectation of the absolute value of each of the last two terms converges towards 0 as $n\rightarrow\infty$ since $|h_{xxx}|\leq M$ and $0\leq p_n(K+1)/p_n(K)\leq 2$.

 We now analyse the fourth one (starting from the bottom). As $n\rightarrow\infty$,
\begin{align*}
   &\E\left[\left|\ind(K=1)\frac{1-\tau}{2(A+1)}h_{xx}(R^K,X_1^K)\left(\frac{p_n(K+1) }{p_n(K)}-2\right)\right|\right] \cr
   &\qquad\leq \frac{M(1-\tau)}{A+1} \times \E\left[\left|\ind(K=1)\right|\right] =\frac{M(1-\tau)}{A+1} \times \Prob(K=1)\rightarrow 0,
  \end{align*}
  using $\abs{h_{xx}}\leq M$ and $0\leq \abs{p_n(K+1)/p_n(K)-2}\leq 2$ in the first inequality. Proposition~\ref{prop_bound_dom} in Section~\ref{lemmas} is then used to conclude that $\Prob(K=1)\rightarrow 0$ as $n\rightarrow\infty$. The proof for the third term (starting from the bottom) is similar.

  Applying Lemmas~\ref{lemma_5} to \ref{lemma_8} from Section~\ref{lemmas}, each of the remaining terms is seen to converge towards 0 in $L^1$ as $n\rightarrow\infty$, and thus
  $$
   \E\left[\abs{\varphi_{2,n}(t)+\varphi_{3,n}(t)-G_2h(\mathbf{Z}_{1,2}^n(t))}\right]\rightarrow 0 \text{ as } n\rightarrow\infty.
   $$

\section{Results Used in the Proof of Theorem~\ref{main_result}}\label{lemmas}

 \begin{lemma}\label{lemma_5}
  As $n\rightarrow\infty$, we have 
  \begin{align*}
   &\E\left[\left|\ind(2\leq K\leq \lfloor \sqrt{n}\log n\rfloor-1)\frac{1-\tau}{2(A+1)}h_{xx}(R^K,X_1^K)\left(\frac{p_n(K+1) }{p_n(K)}-1\right)\right|\right]\rightarrow 0.
  \end{align*}
 \end{lemma}

 \begin{proof}[Proof of Lemma \ref{lemma_5}]
  First,
  \begin{align*}
   &\E\left[\left|\ind(2\leq K\leq \lfloor \sqrt{n}\log n\rfloor-1)\frac{1-\tau}{2(A+1)}h_{xx}(R^K,X_1^K)\left(\frac{p_n(K+1) }{p_n(K)}-1\right)\right|\right] \cr
  &\qquad\leq \frac{M(1-\tau)}{2(A+1)}\E\left[\ind(2\leq K\leq \lfloor \sqrt{n}\log n\rfloor-1) \left|\frac{p_n(K+1) }{p_n(K)}-1\right|\right],
  \end{align*}
  because $\abs{h_{xx}}\leq M$.

 Considering the case where $\lfloor \sqrt{n}\log n\rfloor$ is odd, we have
  \begin{align*}
   &\E\left[\ind(2\leq K\leq \lfloor \sqrt{n}\log n\rfloor-1)\abs{\frac{p_n(K+1) }{p_n(K)}-1}\right] \cr
  &\quad=\E\left[\ind\left(\frac{\lfloor \sqrt{n}\log n\rfloor+1}{2}\leq K\leq \lfloor \sqrt{n}\log n\rfloor-1\right)\abs{\frac{p_n(K+1) }{p_n(K)}-1}\right] \cr
  &\qquad+\E\left[\ind\left(2\leq K\leq \frac{\lfloor \sqrt{n}\log n\rfloor-1}{2}\right)\abs{\frac{p_n(K+1) }{p_n(K)}-1}\right].
  \end{align*}
 We analyse each term separately. The first one corresponds to the case where $K$ is at the mode or at its right. Therefore, $p_n(K+1)/p_n(K)=a_{K,n}$ and

  \begin{align*}
   &\E\left[\ind\left(\frac{\lfloor \sqrt{n}\log n\rfloor+1}{2}\leq K\leq \lfloor \sqrt{n}\log n\rfloor-1\right)\abs{\frac{p_n(K+1) }{p_n(K)}-1}\right] \cr
   &\quad=\E\left[\ind\left(\frac{\lfloor \sqrt{n}\log n\rfloor+1}{2}\leq K\leq \lfloor \sqrt{n}\log n\rfloor-1\right)\abs{a_{K,n}-1}\right] \cr
   &\quad=\E\left[\ind\left(\frac{\lfloor \sqrt{n}\log n\rfloor+1}{2}\leq K\leq \lfloor \sqrt{n}\log n\rfloor-1\right)\frac{\abs{K-\lfloor \sqrt{n}\log n\rfloor/2}}{n}\right] \cr
    &\quad\leq \frac{\lfloor \sqrt{n}\log n\rfloor/2-1}{n} \leq \frac{\log n}{2\sqrt{n}} \rightarrow 0,
  \end{align*}
  because $a_{K,n}=1-b_{K,n}/\sqrt{n}$ and $b_{K,n}=\abs{K-\lfloor \sqrt{n}\log n\rfloor/2}/\sqrt{n}$. We now study the case where $K$ is at the left of the mode. Therefore, $p_n(K+1)/p_n(K)=a_{K,n}^{-1}$ and
  \begin{align*}
   &\E\left[\ind\left(2\leq K\leq \frac{\lfloor \sqrt{n}\log n\rfloor-1}{2}\right)\abs{\frac{p_n(K+1) }{p_n(K)}-1}\right] \cr
    &\quad=\E\left[\ind\left(2\leq K\leq \frac{\lfloor \sqrt{n}\log n\rfloor-1}{2}\right)\abs{a_{K,n}^{-1}-1}\right]\cr
    &\quad=\E\left[\ind\left(2\leq K\leq \frac{\lfloor \sqrt{n}\log n\rfloor-1}{2}\right)\abs{\frac{b_{K,n}}{\sqrt{n}-b_{K,n}}}\right]\leq \frac{(\log n)/2}{\sqrt{n}-(\log n)/2}\rightarrow0,
  \end{align*}
  using similar mathematical arguments as above and the fact that $1/(2\sqrt{n})\leq b_{K,n}\leq (\log n)/2$ when $2\leq K\leq (\lfloor \sqrt{n}\log n\rfloor-1)/2$. The proof for the case where $\lfloor \sqrt{n}\log n\rfloor$ is even is similar.
 \end{proof}

 \begin{lemma}\label{lemma_6}
  As $n\rightarrow\infty$, we have 
  $$
   \E\left[\left|\ind(K=1)\frac{1-\tau}{A+1}h_x(R^K,X_1^K)\left(\sqrt{n}\times \frac{p_n(K+1) }{p_n(K)}+R^K\right)\right|\right]\rightarrow 0,
  $$
  and
  $$
   \E\left[\left|\ind(K=\lfloor\sqrt{n}\log n\rfloor)\frac{1-\tau}{A+1}h_x(R^K,X_1^K)\left(\sqrt{n}-R^K\right)\right|\right]\rightarrow 0.
  $$
 \end{lemma}

 \begin{proof}[Proof of Lemma \ref{lemma_6}]
  Using Proposition \ref{prop_bound_dom}, $\abs{h_x}\leq M$, $0\leq p_n(K+1)/p_n(K)\leq 2$ and $R^K:=(K-\lfloor \sqrt{n}\log n\rfloor/2)/\sqrt{n}$, we have
  \begin{align*}
  &\E\left[\left|\ind(K=1)\frac{1-\tau}{A+1}h_x(R^K,X_1^K)\left(\sqrt{n}\times \frac{p_n(K+1) }{p_n(K)}+R^K\right)\right|\right] \cr
  &\qquad\leq\frac{(1-\tau) M}{A+1}\left(2\sqrt{n}+\frac{ \log n }{2}\right)\Prob(K=1)\rightarrow 0,
  \end{align*}
  since
  \begin{align*}
   \ind(K=1) \abs{ \sqrt{n}\times\frac{p_n(K+1) }{p_n(K)}+R^K}&\leq \ind(K=1)\left(2\sqrt{n}+\abs{\frac{1-\lfloor \sqrt{n}\log n\rfloor /2 }{\sqrt{n}}}\right) \cr
                                                        &\leq\ind(K=1)\left( 2\sqrt{n}+\frac{\log n }{2}\right). \hspace{18mm}
  \end{align*}
   The proof that
   $$
    \E\left[\left|\ind(K=\lfloor\sqrt{n}\log n\rfloor)\frac{1-\tau}{A+1}h_x(R^K,X_1^K)\left(\sqrt{n}-R^K\right)\right|\right]\rightarrow 0
   $$
   is similar.
 \end{proof}

 \begin{lemma}\label{lemma_8}
  As $n\rightarrow\infty$, we have 
  \begin{align*}
   &\E\left[\left|\ind(2\leq K\leq \lfloor \sqrt{n}\log n\rfloor-1)\frac{1-\tau}{A+1}h_x(R^K,X_1^K)\left[\sqrt{n}\left(\frac{p_n(K+1) }{p_n(K)}-1\right)+R^K\right]\right|\right]\rightarrow 0.
  \end{align*}
 \end{lemma}

 \begin{proof}[Proof of Lemma \ref{lemma_8}]
   First,
   \begin{align*}
   &\E\left[\left|\ind(2\leq K\leq \lfloor \sqrt{n}\log n\rfloor-1)\frac{1-\tau}{A+1}h_x(R^K,X_1^K)\left[\sqrt{n}\left(\frac{p_n(K+1) }{p_n(K)}-1\right)+R^K\right]\right|\right] \cr
   &\quad\leq\frac{(1-\tau)M}{A+1}\E\left[\ind(2\leq K\leq \lfloor \sqrt{n}\log n\rfloor-1)\left|\sqrt{n}\left(\frac{p_n(K+1) }{p_n(K)}-1\right)+R^K\right|\right],
  \end{align*}
   because $\abs{h_{x}}\leq M$. Considering the case where $\lfloor \sqrt{n}\log n\rfloor$ is odd, we have
  \begin{align*}
   &\E\left[\ind(2\leq K\leq \lfloor \sqrt{n}\log n\rfloor-1)\left|\sqrt{n}\left(\frac{p_n(K+1) }{p_n(K)}-1\right)+R^K\right|\right] \cr
   &\quad=\E\left[\ind\left(\frac{\lfloor \sqrt{n}\log n\rfloor+1}{2}\leq K\leq \lfloor \sqrt{n}\log n\rfloor-1\right)\left|\sqrt{n}\left(\frac{p_n(K+1) }{p_n(K)}-1\right)+R^K\right|\right] \cr
  &\qquad+\E\left[\ind\left(2\leq K\leq \frac{\lfloor \sqrt{n}\log n\rfloor-1}{2}\right)\left|\sqrt{n}\left(\frac{p_n(K+1) }{p_n(K)}-1\right)+R^K\right|\right].
  \end{align*}
  We analyse each term separately. The first one corresponds to the case where $K$ is at the mode or at its right. Therefore, $p_n(K+1)/p_n(K)=a_{K,n}$ and
  \begin{align*}
   &\E\left[\ind\left(\frac{\lfloor \sqrt{n}\log n\rfloor+1}{2}\leq K\leq \lfloor \sqrt{n}\log n\rfloor-1\right)\left|\sqrt{n}\left(\frac{p_n(K+1) }{p_n(K)}-1\right)+R^K\right|\right] \cr
   &\quad=\E\left[\ind\left(\frac{\lfloor \sqrt{n}\log n\rfloor+1}{2}\leq K\leq \lfloor \sqrt{n}\log n\rfloor-1\right)\left|\sqrt{n}\left(a_{K,n} -1\right)+R^K\right|\right] \cr
   &\quad= \E\left[\ind\left(\frac{1}{2\sqrt{n}}\leq R^K\leq  \frac{\lfloor\sqrt{n}\log n\rfloor-2}{2\sqrt{n}}\right)\abs{-b_{K,n}+R^K}\right]=0,
  \end{align*}
  because $a_{K,n}=1-b_{K,n}/\sqrt{n}$ and $b_{K,n}=\abs{K-\lfloor \sqrt{n}\log n\rfloor/2}/\sqrt{n}=R^K$ when $R^K\geq0$ ($R^K:=(K-\lfloor \sqrt{n}\log n\rfloor/2)/\sqrt{n}$). We now study the case where $K$ is at the left of the mode. Therefore, $p_n(K+1)/p_n(K)=a_{K,n}^{-1}$ and
  \begin{align*}
   &\E\left[\ind\left(2\leq K\leq \frac{\lfloor \sqrt{n}\log n\rfloor-1}{2}\right)\left|\sqrt{n}\left(\frac{p_n(K+1) }{p_n(K)}-1\right)+R^K\right|\right] \cr
   &\quad=\E\left[\ind\left(2\leq K\leq \frac{\lfloor \sqrt{n}\log n\rfloor-1}{2}\right)\left|\sqrt{n}\left(a_{K,n}^{-1}-1\right)+R^K\right|\right] \cr
   &\quad=\E\left[\ind\left(2\leq K\leq \frac{\lfloor \sqrt{n}\log n\rfloor-1}{2}\right)\left|\frac{\sqrt{n}}{\sqrt{n}-b_{K,n}}\times b_{K,n}+R^K\right|\right] \cr
  &\quad=\E\left[\ind\left(\frac{4-\lfloor \sqrt{n}\log n\rfloor}{2\sqrt{n}}\leq R^K \leq \frac{-1}{2\sqrt{n}}\right)\times-R^K\left|\frac{\sqrt{n}}{\sqrt{n}-b_{K,n}}-1\right|\right] \cr
  &\quad\leq \E\left[\ind\left(\frac{4-\lfloor \sqrt{n}\log n\rfloor}{2\sqrt{n}}\leq R^K \leq \frac{-1}{2\sqrt{n}}\right)\frac{\log n}{2}\times\frac{b_{K,n}}{\sqrt{n}-b_{K,n}}\right] \cr
  &\quad\leq \frac{\log n}{2}\times\frac{(\log n)/2}{\sqrt{n}-(\log n)/2}\rightarrow 0,
  \end{align*}
  using similar mathematical arguments as above, the fact that $b_{K,n}=-R^K$ when $R^K<0$, and that $1/(2\sqrt{n})\leq-R^K\leq (\log n)/2$ when $(4-\lfloor \sqrt{n}\log n\rfloor)/$ $(2\sqrt{n})\leq R^K \leq -1/(2\sqrt{n})$. The proof for the case where $\lfloor \sqrt{n}\log n\rfloor$ is even is similar.
 \end{proof}

 \begin{prop}\label{prop_bound_dom}
 The random variable $K$ with PMF $p$ defined in Section~\ref{sec_context_weak} is such that, for all $\rho\in\re$,
 $$
  \lim_{n\rightarrow\infty}n^\rho\Prob(K=1)=\lim_{n\rightarrow\infty}n^\rho\Prob(K=\lfloor \sqrt{n}\log n\rfloor)=0.
 $$
 \end{prop}

 \begin{proof}[Proof of Proposition~\ref{prop_bound_dom}]
  Consider the case where $\lfloor \sqrt{n}\log n\rfloor$ is even. Using equation~(\ref{eqn_prob_k_even}), we have
  $$
   p_n(1)=p_n(\lfloor \sqrt{n}\log n\rfloor)=p_n\left(\frac{\lfloor \sqrt{n}\log n\rfloor}{2}\right)\prod_{i=1}^{\frac{\lfloor \sqrt{n}\log n\rfloor}{2}-1}\left(1-\frac{i}{n}\right).
  $$
  In the proof of Proposition~\ref{prop_z_1_to_normal} in Section~\ref{sec_proofs_props}, we show that $p_n(\lfloor \sqrt{n}\log n\rfloor/2)\rightarrow 1/\sqrt{2\pi}$ as $n\rightarrow\infty$. Also, using the fact that $1-x\leq\exp\{-x\}$ for all $x\in\re$, we have for all $\rho\in\re$
  \begin{align*}
   n^\rho\prod_{i=1}^{\frac{\lfloor \sqrt{n}\log n\rfloor}{2}-1} \left(1-\frac{i}{n}\right)&\leq  n^\rho \exp\left\{-\sum_{i=1}^{\frac{\lfloor \sqrt{n}\log n\rfloor}{2}-1} \frac{i}{n}\right\} \cr
   &=n^\rho\exp\left\{-\frac{1}{2}\left(\frac{\frac{\lfloor \sqrt{n}\log n\rfloor}{2}-1}{\sqrt{n}}\right)^2 -\frac{\frac{\lfloor \sqrt{n}\log n\rfloor}{2}-1}{2n}\right\}
  \end{align*}
  \begin{align*}
   &\hspace{20mm}\leq n^\rho \exp\left\{-\frac{1}{2}\left(\frac{\lfloor \sqrt{n}\log n\rfloor}{2\sqrt{n}}-1\right)^2\right\}\rightarrow 0,
  \end{align*}
  as $n\rightarrow\infty$. Similarly, we can show the result for the case where $\lfloor \sqrt{n}\log n\rfloor$ is odd.
 \end{proof}

\section{Proofs of Propositions~\ref{prop_z_1_to_normal} and \ref{prop_prob_acc_ajout}}\label{sec_proofs_props}

\begin{proof}[Proof of Proposition~\ref{prop_z_1_to_normal}]
  \quad The random variable $Z_1^n(t)$ is defined as $(K(\lfloor nt \rfloor)-\lfloor \sqrt{n}\log n\rfloor/2)/\sqrt{n}$, and $K(\lfloor nt \rfloor)\sim p_n$ for all $t$ (see Section~\ref{sec_context_weak} for the assumptions on $p_n$). Therefore, to simplify the notation, the time index is omitted for the rest of the proof. Consider the constant $z<0$ and the case where $\lfloor \sqrt{n}\log n\rfloor$ is even. We have
  \begin{align*}
   \Prob((K-\lfloor \sqrt{n}\log n\rfloor/2)/\sqrt{n}\leq z)=\Prob(K\leq z\sqrt{n}+\lfloor \sqrt{n}\log n\rfloor/2) &= \sum_{k=\lceil -z\sqrt{n}\rceil}^{\lfloor \sqrt{n}\log n\rfloor/2-1} p_n(\lfloor \sqrt{n}\log n\rfloor/2-k) \cr
   &\hspace{-8mm}=p_n\left(\frac{\lfloor \sqrt{n} \log n\rfloor}{2}\right)\sum_{k=\lceil -z\sqrt{n}\rceil}^{\frac{\lfloor \sqrt{n}\log n\rfloor}{2}-1} \prod_{i=1}^k \left(1-\frac{i}{n}\right),
  \end{align*}
  using $\lfloor z\sqrt{n}+\lfloor \sqrt{n}\log n\rfloor/2\rfloor=\lfloor z\sqrt{n}\rfloor+\lfloor \sqrt{n}\log n\rfloor/2= \lfloor \sqrt{n}\log n\rfloor/2-\lceil- z\sqrt{n}\rceil$ in the second equality ($\lceil \cdot \rceil$ is the ceiling function), and equation~(\ref{eqn_prob_k_even}) in the last equality. The sum above is well-defined if $n\geq \exp(-2z+5)$ and we select $n$ large enough to ensure this. Using again equation~(\ref{eqn_prob_k_even}) and the fact that $\sum_{k=1}^{\lfloor\sqrt{n}\log n\rfloor} p_n(k)=1$, we have
  $$
   p_n\left(\frac{\lfloor \sqrt{n} \log n\rfloor}{2}\right)=\left(2\left(1+\sum_{k=1}^{\lfloor \sqrt{n}\log n\rfloor/2-1} \prod_{i=1}^k \left(1-\frac{i}{n}\right)\right)\right)^{-1}.
  $$
  Therefore,
  $$
   \Prob\left(\frac{K-\lfloor \sqrt{n}\log n\rfloor/2}{\sqrt{n}}\leq z\right)=\frac{(1/\sqrt{n})\sum_{k=\lceil -z\sqrt{n}\rceil}^{\lfloor \sqrt{n}\log n\rfloor/2-1} \prod_{i=1}^k (1-i/n)}{\frac{2}{\sqrt{n}}+\frac{2}{\sqrt{n}}\sum_{k=1}^{\lfloor \sqrt{n}\log n\rfloor/2-1} \prod_{i=1}^k (1-i/n)}.
  $$
  Using the fact that $1-x\leq\exp\{-x\}$ for all $x\in\re$, we have
  $$
   \prod_{i=1}^k \left(1-\frac{i}{n}\right)\leq   \exp\left\{-\sum_{i=1}^k \frac{i}{n}\right\}=\exp\left\{-\frac{1}{2}\left(\frac{k}{\sqrt{n}}\right)^2 -\frac{k}{2n}\right\}.
  $$
  In addition, for all $\delta>0$, there exists $\epsilon>0$ such that $\exp\{-(1+\delta)x\}\leq 1-x$ for $0\leq x<\epsilon$. Therefore, since $0\leq i/n\leq \lfloor\sqrt{n}\log n\rfloor/(2n)-1/n\leq \log n/(2\sqrt{n})\rightarrow 0$ as $n\rightarrow\infty$ when $1\leq i\leq k\leq \lfloor \sqrt{n}\log n\rfloor/2-1$, for all $\delta>0$, there exists a constant $N>0$ such that for all $n\geq N$,
  $$
   \prod_{i=1}^k \left(1-\frac{i}{n}\right)\geq\exp\left\{-\frac{(1+\delta)}{2}\left(\frac{k}{\sqrt{n}}\right)^2 -\frac{k(1+\delta)}{2n}\right\}.
  $$
  The objective is to use a ``Riemann sum'' argument, where the length of the subintervals of the partition is $1/\sqrt{n}$, to study the asymptotic behaviour of the numerator and denominator of $\Prob((K-\lfloor \sqrt{n}\log n\rfloor/2)/\sqrt{n}\leq z)$. More precisely, we now prove that the numerator of $\Prob((K-\lfloor \sqrt{n}\log n\rfloor/2)/\sqrt{n}\leq z)$ converges towards $\int_{-z}^\infty \exp(-x^2/2) dx$ and that the denominator converges towards $\int_{-\infty}^\infty \exp(-x^2/2) dx=\sqrt{2\pi}$. To achieve this, we use Lebesgue's dominated convergence theorem. First, we rewrite the numerator as
  $$
   \frac{1}{\sqrt{n}}\sum_{k=\lceil -z\sqrt{n}\rceil}^{\frac{\lfloor \sqrt{n}\log n\rfloor}{2}-1} \prod_{i=1}^k \left(1-\frac{i}{n}\right)= \int_{-z}^\infty \sum_{k=\lceil -z\sqrt{n}\rceil}^{\frac{\lfloor \sqrt{n}\log n\rfloor}{2}-1} \prod_{i=1}^k \left(1-\frac{i}{n}\right)\ind_{\left[\frac{k}{\sqrt{n}},\frac{k+1}{\sqrt{n}}\right)}(x) dx.
  $$
  Now, we analyse the integrand. For all $x\in(-z,\infty)$ and for large enough $n$, there exists a unique $k'\in\{\lceil -z\sqrt{n}\rceil,\ldots, \lfloor \sqrt{n}\log n\rfloor/2-1\}$ with $\ind_{[k'/\sqrt{n},(k'+1)/\sqrt{n})}(x)=1$. Also, $0\leq x-k'/\sqrt{n}<1/\sqrt{n}$, which implies that $k'/\sqrt{n}\rightarrow x$ as $n\rightarrow\infty$. Consequently, using the upper bound and the lower bound on $\prod_{i=1}^k (1-i/n)$,
  $$
   \sum_{k=\lceil -z\sqrt{n}\rceil}^{\frac{\lfloor \sqrt{n}\log n\rfloor}{2}-1} \prod_{i=1}^k \left(1-\frac{i}{n}\right)\ind_{\left[\frac{k}{\sqrt{n}},\frac{k+1}{\sqrt{n}}\right)}(x)=\prod_{i=1}^{k'} \left(1-\frac{i}{n}\right)\rightarrow \exp\{-x^2/2\},
  $$
  as $n\rightarrow\infty$, because $k'/n\leq \lfloor\sqrt{n}\log n\rfloor/(2n)-1/n\leq \log n/(2\sqrt{n})\rightarrow 0$. Now, we prove that the integrand is bounded by an integrable function that does not depend on $n$. For all $x\in(-z,\infty)$,
  \begin{align*}
    \prod_{i=1}^k \left(1-\frac{i}{n}\right)\ind_{\left[\frac{k}{\sqrt{n}},\frac{k+1}{\sqrt{n}}\right)}(x)&\leq \exp\left\{-\frac{1}{2}\left(\frac{k}{\sqrt{n}}\right)^2  -\frac{k}{2n}\right\}\ind_{\left[\frac{k}{\sqrt{n}},\frac{k+1}{\sqrt{n}}\right)}(x) \leq \exp\left\{-\frac{1}{2}(x-1)^2\right\}\ind_{\left[\frac{k}{\sqrt{n}},\frac{k+1}{\sqrt{n}}\right)}(x),
  \end{align*}
  using the upper bound on $\prod_{i=1}^k (1-i/n)$ in the first inequality, and then $x\leq (k+1)/\sqrt{n}\leq (k/\sqrt{n})+1$. As a result,
  \begin{align*}
   \sum_{k=\lceil -z\sqrt{n}\rceil}^{\frac{\lfloor \sqrt{n}\log n\rfloor}{2}-1} \prod_{i=1}^k \left(1-\frac{i}{n}\right)\ind_{\left[\frac{k}{\sqrt{n}},\frac{k+1}{\sqrt{n}}\right)}(x)&\leq \exp\left\{-\frac{1}{2}(x-1)^2\right\} \sum_{k=\lceil -z\sqrt{n}\rceil}^{\frac{\lfloor \sqrt{n}\log n\rfloor}{2}-1}\ind_{\left[\frac{k}{\sqrt{n}},\frac{k+1}{\sqrt{n}}\right)}(x)\cr
   &= \exp\left\{-\frac{1}{2}(x-1)^2\right\} \ind_{\left[\frac{\lceil -z\sqrt{n}\rceil}{\sqrt{n}},\frac{\lfloor \sqrt{n}\log n\rfloor/2}{\sqrt{n}}\right)}(x)\leq \exp\left\{-\frac{1}{2}(x-1)^2\right\},
  \end{align*}
  which is integrable. Similarly, we can prove that the denominator of $\Prob((K-\lfloor \sqrt{n}\log n\rfloor/2)/$ $\sqrt{n}\leq z)$ converges towards $\int_{-\infty}^\infty \exp(-x^2/2) dx=\sqrt{2\pi}$, and we can show the result for $z\geq 0$ and for the case where $\lfloor \sqrt{n}\log n\rfloor$ is odd.
  \end{proof}

  \begin{proof}[Proof of Proposition~\ref{prop_prob_acc_ajout}]
    The random variables $K(m)$ and $U(m+1)$ are independent and such that $K(m)\sim p_n$ and $U(m+1)\sim q$ for all $m\in\na$ (see Section~\ref{sec_context_weak} for the assumptions on $p_n$ and $q$). Therefore, to simplify the notation, the time index is omitted for the rest of the proof. As explained in Section~\ref{sec_proof_finite_dim},
    $$
     \E\left[1 \wedge \frac{f(U)p_n(K+1)}{q(U)p_n(K)A}\right]=\E\left[ \frac{f(U)p_n(K+1)}{q(U)p_n(K)A}\right],
  $$
  and $\E[f(U)/q(U)]=1$. Finally, using Proposition~\ref{prop_bound_dom}, we have
  \begin{align*}
   &\frac{1}{A}\E\left[ \frac{p_n(K+1)}{p_n(K)}\right]=\frac{1}{A}\sum_{k=1}^{\lfloor \sqrt{n}\log n\rfloor-1}p_n(k+1)=\frac{1}{A}(1-p_n(1))\rightarrow \frac{1}{A} \text{ as } n\rightarrow\infty.
  \end{align*}
  \end{proof}



\section{Acknowledgements}

The authors acknowledge support from the NSERC (Natural Sciences and Engineering  Research  Council  of  Canada),  the  FRQNT  (Le Fonds de recher\-che du Qu\'{e}bec - Nature et technologies) and the SOA (Society of Actuaries). The authors thank an anonymous referee and an associate editor for their helpful comments.

\bibliographystyle{model1b-num-names}
\bibliography{reference}

\end{document}